\documentclass[11pt]{article}%
\usepackage[dvipsnames]{xcolor}

\usepackage{amsfonts}
\usepackage{fancyhdr}
\usepackage[hidelinks]{hyperref}
\usepackage[a4paper, top=2.54cm, bottom=2.54cm, left=2.54cm, right=2.54cm]%
{geometry}
\usepackage{hyperref}
\usepackage{calrsfs}

\usepackage{times}
\usepackage{amsmath}
\usepackage{mathabx}
\usepackage{changepage}
\usepackage{amssymb}
\usepackage{graphicx}%
\newtheorem{theorem}{Theorem}[section]
\newtheorem{corollary}[theorem]{Corollary}
\newtheorem{definition}[theorem]{Definition}
\newtheorem{lemma}[theorem]{Lemma}
\newtheorem{proposition}[theorem]{Proposition}
\newtheorem{observation}[theorem]{Observation}

\newtheorem{remark}[theorem]{Remark}

\usepackage{mathtools}

\DeclareMathOperator*{\argmin}{arg\,min}
\DeclarePairedDelimiter\ceil{\lceil}{\rceil}

\newenvironment{proof}[1][Proof]{\textbf{#1.} }{\ \rule{0.5em}{0.5em}}
\usepackage{breqn}
\usepackage{physics}
\usepackage{bbm}
\usepackage{algorithm2e}
\usepackage{cases}    

\begin{document}

\title{Simple and Optimal Online Contention Resolution Schemes for $k$-Uniform Matroids} 
\author{
  Atanas Dinev\\
  Massachusetts Institute of Technology, Cambridge MA, USA
  \and
  S. Matthew Weinberg\\
  Princeton University, Princeton, USA\\
}

\date{\today}
\maketitle
\begin{abstract}

We provide a simple $(1-O(\frac{1}{\sqrt{k}}))$-selectable Online Contention Resolution Scheme for $k$-uniform matroids against a fixed-order adversary. If $A_i$ and $G_i$ denote the set of selected elements and the set of realized active elements among the first $i$ (respectively), our algorithm selects with probability $1-\frac{1}{\sqrt{k}}$ any active element $i$ such that $|A_{i-1}| + 1 \leq (1-\frac{1}{\sqrt{k}})\cdot \mathbb{E}[|G_i|]+\sqrt{k}$. This implies a $(1-O(\frac{1}{\sqrt{k}}))$ prophet inequality against fixed-order adversaries for $k$-uniform matroids that is considerably simpler than previous algorithms~\cite{Alaei11, AKW14, jmz22}.

We also prove that no OCRS can be $(1-\Omega(\sqrt{\frac{\log k}{k}}))$-selectable for $k$-uniform matroids against an almighty adversary. This guarantee is matched by the (known) simple greedy algorithm that accepts every active element with probability $1-\Theta(\sqrt{\frac{\log k}{k}})$~\cite{hks07}.

\end{abstract}
\addtocounter{page}{-1}
\newpage
\section{Introduction}
\noindent\textbf{Background: OCRSs.} Online contention resolution schemes (OCRS) are a broadly applicable rounding technique for online selection problems~\cite{fsz16}. These are problems in which an algorithm makes irrevocable decisions for whether to select elements arriving online, often subject to combinatorial constraints. Offline, the algorithm knows a distribution over which elements will be ``active.'' Online, elements are revealed to be active or inactive one at a time, and the algorithm must immediately and irrevocably decide whether to accept an active element (inactive elements must be rejected). There are feasibility constraints $\mathcal{F}$ on which elements can be simultaneously accepted. An OCRS for a class of instances is $c$-selectable if it guarantees that every element is selected with probability at least $c$, conditioned on being active. See Definition~\ref{ocrs_def} for a formal definition.

Online contention resolution schemes have direct applications to prophet inequalities~\cite{fsz16}. In a prophet inequality, a gambler knows the distribution of a sequence of independent random variables $X_1,\ldots, X_n$. Online, each random variable will be sampled one at a time and revealed to the gambler, who immediately and irrevocably decides whether to accept the element, subject to feasibility constraints $\mathcal{F}$. The gambler's goal is to maximize the expected sum of weights of accepted elements, and a prophet inequality compares the ratio of the gambler's expected performance to that of a prophet (who knows all random variables before making decisions). Seminal work of Krengel, Sucheston, and Garling establishes a tight $1/2$-approximation for the single-choice prophet inequality, and seminal work of Samuel-Cahn shows that the same result can be achieved with an especially simple thresholding algorithm~\cite{KrengelS78,Samuel-Cahn84}.

In their work introducing OCRSs, Feldman, Svensson, and Zenklusen prove that a $c$-selectable OCRS implies a $c$-approximation for the corresponding prophet inequality~\cite{fsz16}. In fact, a $c$-selectable OCRS provides a $c$-approximation even to the ex ante relaxation, and Lee and Singla show that OCRSs are \emph{equivalent} to ex ante prophet inequalities~\cite{ls18}.\\

\noindent\textbf{$k$-Uniform Matroids.} $k$-uniform matroids are a canonical set of feasibility constraints: any set of up to $k$ elements can be selected. Here, a $(1-O(\sqrt{\frac{\log k}{k}}))$-approximate prophet inequality (whose analysis implicitly extends to a $(1-O(\sqrt{\frac{\log k}{k}}))$-selectable OCRS) is first developed in~\cite{hks07}.~\cite{Alaei11} later develops a $(1-O(\frac{1}{\sqrt{k}}))$-approximation (which also implies a $(1-O(\frac{1}{\sqrt{k}}))$-selectable OCRS), which is tight.~\cite{AKW14} further shows how to achieve the same $(1-O(\frac{1}{\sqrt{k}}))$-approximation using a single sample from each distribution (but their work does not imply any OCRS).~\cite{jmz22} tightens the analysis of~\cite{Alaei11} to nail down exactly the optimal achievable prophet inequality for all $k$ (and this same analysis applies to the implied OCRS). We overview in Section~\ref{our_results_prior_work} these results in more detail, and in particular clarify against what kind of adversary (who selects the order in which the elements are revealed) the guarantees hold.

While the optimal competitive ratio has been known for a decade, and recently tightened to even nail down the precise constants, these algorithms are significantly more complex than Samuel-Cahn's elegant algorithm for the single-choice prophet inequality. For example,~\cite{Alaei11,jmz22} both require to solve and analyze a mathematical program in order to accept elements with precisely the correct probability. The rehearsal algorithm of~\cite{AKW14} is perhaps simpler, but still requires several lines of pseudocode, and some care with minor details.\\ 

\noindent\textbf{Main Result: A Simple, Optimal OCRS.} Our main result is a significantly simpler OCRS/prophet inequality for $k$-uniform matroids that still achieves the optimal guarantee of $1-O(\frac{1}{\sqrt{k}})$. Of course, it is still not nearly as simple as Samuel-Cahn's single-choice prophet inequality, but a full description fits in two sentences, and the complete analysis is just a few pages.\footnote{The proof does require connecting our algorithm to a random walk, and then analyzing properties of that random walk. But, the proof requires minimal calculations.} Our OCRS simply denotes by $A_i$ the set of elements it has selected amongst the first $i$ and by $G_i$ the set of realized active elements amongst the first $i$. Then, when processing element $i$, we select $i$ with probability $1-\frac{1}{\sqrt{k}}$ if and only if $i$ is active and $|A_{i-1}| + 1 \leq (1-\frac{1}{\sqrt{k}})\cdot \mathbb{E}[|G_i|] + \sqrt{k}$ (otherwise, we discard). Intuitively, our OCRS selects an element if, so far, the number of selected elements does not exceed the expected number of active elements by too much. To turn our OCRS into a prophet inequality, simply let $T$ denote the unique value such that $\sum_i \Pr[X_i > T] =k$.\footnote{If the distributions have point-masses, smooth them out by adding a uniformly random draw from $[0,\varepsilon]$ for arbitrarily small $\varepsilon$.} Then, declare $X_i$ to be active if and only if $X_i > T$ and plug this into our OCRS. Compared to prior optimal algroithms for the same setting, our algorithm has the advantage that it is very simple to implement since it does not require solving a complicated dynamic/linear program. (See Section \ref{our_results_prior_work}) We state our algorithm precisely, and prove that it is $(1-O(\frac{1}{\sqrt{k}}))$-selectable against a fixed-order adversary in Section~\ref{simple_ocrs_section}.\footnote{A fixed-order adversary sets the order to reveal the elements offline, and based only on the distributions.}  \\

\noindent\textbf{An Impossibility for OCRS against almighty adversary.} While the fixed-order adversary is standard in the prophet inequality literature, it is also important to explore the extent to which these same guarantees can hold against an almighty adversary.\footnote{An almighty adversary sets the order to reveal online, and with full knowledge of all random variables and all past decisions of the algorithm.} For example, the prophet inequality of~\cite{AKW14} holds against an almighty adversary, but does not imply an OCRS. Our second result shows that this is for good reason: no OCRS can guarantee a selectability better than $1-\Omega(\sqrt{\frac{\log k}{k}})$ against an almighty adversary. We state and prove this in Section~\ref{opt_guarantee_against_almigthy}.

\subsection{Detailed Discussion of Related Work}\label{our_results_prior_work}

As previously referenced, prior to our work it is already known that the optimal selectability for OCRS and the optimal prophet inequality against a fixed-order adversary is $1-\Theta(\frac{1}{\sqrt{k}})$ (\cite{hks07} proves the impossibility, and~\cite{Alaei11} designs the first algorithm matching it).~\cite{AKW14} designs a prophet inequality that achieves the same $(1-O(\frac{1}{\sqrt{k}}))$-approximation against an almighty adversary, but this does not imply an OCRS. The analysis in~\cite{hks07} implies an extremely simple OCRS (accept every active element with probability $1-\Theta(\sqrt{\frac{\log k}{k}})$) that is $(1-\Theta(\sqrt{\frac{\log k}{k}}))$-selectable against an almighty adversary. We show that this is the best possible guarantee (Theorem~\ref{thm_ocrs_almigthy}). Because we view our main result as a simpler algorithm achieving (asymptotically) the same guarantees as prior work, we now overview these works in greater detail.\\

\noindent \textbf{The $\gamma$-Conservative magician in \cite{Alaei11}.} As previously mentioned, \cite{Alaei11} implies an optimal $(1-O(\frac{1}{\sqrt{k}}))$-selectable OCRS against a fixed-order adversary.~\cite[Definition 2]{Alaei11} describes a $\gamma$-Conservative Magician, which is an algorithm that adaptively computes thresholds $\theta_i$ and accepts an active element on step $i$ if and only if the number of selected elements (or broken wands in the terminology of \cite{Alaei11}) $W_i$ in steps $1, \ldots, i-1$ is less than $\theta_i$. The cumulative distribution function of $W_i$ is computed adaptively at every step through dynamic programming equation. Once the CDF of $W_i$ has been computed, $\theta_i$ is chosen so that the ex-ante probability that $W_i \leq \theta_i$ is at least $\gamma$.~\cite{Alaei11} shows that one can choose $\gamma = 1-\frac{1}{\sqrt{k+3}}$ and $\theta_i \leq k$ for all $i$, which effectively guarantees an OCRS against a fixed-order adversary that is $1-\frac{1}{\sqrt{k+3}}$-selectable~\cite[Theorem 4]{Alaei11}. In comparison to~\cite{Alaei11}, our main result achieves the same asymptotic guarantee, but is considerably simpler (in particular, the analysis requires minimal calculations, and there is no dynamic program). \\

\noindent \textbf{Characterization of the optimal OCRS and prophet inequality for $k$-uniform matroids in \cite{jmz22}.} \cite{jmz22} studies the optimal OCRS for $k$-uniform matroids (i.e. with optimal selection probability $c$). They characterize the optimal OCRS for $k$-uniform matroids as the solution to a linear program. Then, using a differential equation, they show that this optimal solution corresponds to a $\gamma^{*}_k$-Conservative magician, where $\gamma^{*}_k > 1-\frac{1}{\sqrt{k+3}}$. \cite{jmz22} extend their OCRS guarantees against an online adversary \footnote{An online adversary adaptively decides which elements to reveal next, based on which elements were active. But the online adversary does not know the status of the unrevealed elements (i.e. it has the same information as the algorithm)}.

In comparison to~\cite{jmz22}, our main result achieves the same $1-O(\frac{1}{\sqrt{k}})$ asymptotic guarantee against a fixed-order adversary, but is, again, considerably simpler and does not require solving any mathematical program. Although our analysis holds against the weaker fixed-order adversary, this assumption is sufficient for many popular applications of OCRS and prophet inequalities in online stochastic optimization. \\

\noindent \textbf{The Rehearsal algorithm in \cite{AKW14}.} As previously discussed,~\cite{AKW14} gives an optimal $(1-O(\frac{1}{\sqrt{k}}))$-approximation prophet inequality against an almighty adversary, and even when knowing only a single sample from each distribution, using the \textit{rehearsal algorithm}. Their rehearsal algorithm takes a sample from each distribution, and stores the $k-2\sqrt{k}$ highest samples at $T_1,\ldots, T_{k-2\sqrt{k}}$, then repeats $T_i:=T_{k-2\sqrt{k}}$ for all $i \in [k-2\sqrt{k},k]$. When processing the online element $X_e$, $e$ is accepted if and only if there is an unfilled slot $i$ with $X_e > T_i$. If $X_e$ is accepted, it fills the highest such slot (the slot with the highest threshold). 

Their analysis does not imply an OCRS (indeed, it is not even clear what it would mean to set the thresholds $T_i$ in an OCRS). But, their analysis does hold against an almighty adversary. In comparison, our prophet inequality is simpler, and implies an OCRS. But, our algorithm requires some knowledge of the distributions, rather than just a single sample.\footnote{Turning our OCRS into a prophet inequality requires a value $T$ such that $\sum_i \Pr[X_i > T] \approx k$, and an accurate estimate of $\sum_{j \leq i}\Pr[X_j > T]$ for all $i$. Estimates up to an additive $\sqrt{k}$ with high probability suffice. This can certainly be achieved with polynomially-many samples, but not a single sample.} 

Our analyses have similar flavors: both works connect our algorithms' performance to a random walk. These random walks are quite different (for example, the random walk in~\cite{AKW14} is correlated, and ours is not. Our random walk has non-integral step sizes, while theirs does not), and are used to analyze different algorithms. While there are some coincidental similarities (for example, our Lemma~\ref{lemma_martingale_bounds} is a generalization of their Lemma~10), the core of our proof is simply connecting our algorithm to a random walk, whereas the bulk of their proof is coping with the correlation in their random walk and any associated calculations. \\

\noindent\textbf{Other Related Work.} There is substantial additional work on both prophet inequalities and online contention resolution schemes, subject to various other constraints~\cite{ChawlaHMS10, KleinbergW12, GobelHKSV14, DuttingK15, Rubinstein16, RubinsteinS17, DuettingFKL17,AdamczykW18, GravinW19, AnariNSS19, EzraFGT20, CaramanisDFFLLPPR22, CorreaCFPW22, ssz23, mmg23}. Aside from this, there is not much technical overlap with this works (in particular, a substantial fraction of these works consider richer feasibility constraints, and therefore achieve constant-factor approximations rather than approximations approaching $1$). 

Other works have also considered a special class of static threshold policies for $k$-unit prophet inequalities which set a single threshold and accept any element that exceeds it subject to the feasibility constraint. \cite{hks07} provides a prophet inequality with a static threshold which is a $1-O(\sqrt{\frac{\log(k)}{k}})$ approximation. \cite{cl20} proposes a different static threshold prophet inequality with the same $1-O(\sqrt{\frac{\log(k)}{k}})$ asymptotic guarantee improving the approximation for small $k$. \cite{jmz23} uses a mathematical programming approach to show that the policy in \cite{cl20} is worst-case optimal within all static threshold policies. In contrast to these works, we study the design and limitations of the richer class of adaptive strategies for the more general setting of OCRS.

In the context of \emph{offline} contention resolution schemes (introduced in \cite{cvz14}),~\cite{ks21} shows a simple optimal contention resolution scheme for $k$-uniform matroids, which is 
$(1-\binom{n}{k}(1-\frac{k}{n})^{n+1-k}(\frac{k}{n})^k)$-selectable. Their method is not extendable to the online case because it requires knowing the set of active elements $A$ in advance.

\subsection{Roadmap}
Section \ref{preliminaries_section} follows with preliminaries and definitions. Section \ref{naive_algorithms} is a warmup that rules out optimal selectability via \emph{extremely} simple greedy algorithms. Section \ref{simple_ocrs_section} presents a complete proof of our new $(1-O(\frac{1}{\sqrt{k}}))$-selectable OCRS. Section \ref{opt_guarantee_against_almigthy} contains a proof of the $1-\Omega(\sqrt{\frac{\log(k)}{k}})$ upper bound on the probability of selection of any OCRS against {almighty} adversaries. Section \ref{conclusions_future_work_section} concludes. 

\section{Preliminaries}\label{preliminaries_section}
\subsection{Online contention resolution schemes}

Online contention resolution schemes were first introduced by \cite{fsz16} as a broadly applicable online rounding framework. Suppose we are given a finite ground set of elements $N = \{e_1, \ldots, e_n\}$. Consider a family of feasible sets $\mathcal{F} \subset 2^{N}$ which is downwards-closed (that is, if $I \in \mathcal{F}$ and $J \subseteq I$, then $J \in \mathcal{F}$). Let $$P_{\mathcal{F}} = \text{conv}(\mathbbm{1}_{I}|I \in \mathcal{F}) \in [0,1]^n $$ be the convex hull of all characteristic vectors in $\mathcal{F}$. We will refer to $P_{\mathcal{F}}$ as the polytope corresponding to the family $\mathcal{F}$. 
\begin{definition}[Online contention resolution scheme (OCRS)]\label{ocrs_def}
Consider the following online selection setting. A point $x \in P_{\mathcal{F}}$ is given and let $R(x)$ be a random subset of active elements, where element $e_i$ is active with probability $x_i$. The elements $e \in N$ reveal one by one whether they are active, i.e. $e \in R(x)$, and the decision of the algorithm whether to select an active element is taken irrevocably before the next element is revealed. An OCRS for $P_{\mathcal{F}}$ is an online algorithm that selects a subset $I \subseteq R(x)$ such that $\mathbbm{1}_{I} \in P_{\mathcal{F}}$. 
\end{definition}
Many of the natural OCRS considered in \cite{fsz16} are also greedy. 
\begin{definition}[Greedy OCRS]\label{greedy_ocrs_def}
A greedy OCRS $\pi$ for $P_{\mathcal{F}}$ is an OCRS that for any $x \in P_{\mathcal{F}}$, defines a downwards-closed subfamily of feasible sets $\mathcal{F}_{x} \subseteq \mathcal{F}$ and an element $e$ is selected when it arrives if, together with the already selected elements, the obtained set is in $\mathcal{F}_{x}$
\end{definition}

\indent Our next goal is to define the notion of $c$-selectability, which is a notion of performance of the OCRS. Intuitively, an OCRS is $c$-selectable if for any element $e$ the probability that the OCRS selects $e$ given that it is active is at least $c$, where we desire $c$ to be as large as possible. In order to talk about $c$-selectability in a rigorous way, we need to specify the power of the adversary that chooses the order of the elements revealed to the OCRS in an online fashion (Definition \ref{ocrs_def}). There are three main types of adversaries considered in prior work, which we define below. 

\begin{definition}\label{adv_stength_ocrs} (Strength of adversary) In the setting of Definition \ref{ocrs_def}, there is an underlying adversary which can choose the order in which the elements are revealed to the OCRS. We define three different types of adversaries:\\
\indent (i) \textbf{Offline/Fixed-Order} adversary, which chooses the order of the elements upfront before any are revealed. Such an adversary knows $x$ and the distribution of $R(x)$, but not the realized active elements. \\
\indent (ii) \textbf{Online} adversary, which adaptively chooses next element to reveal using the same information available to the algorithm ($x$, the distribution of $R(x)$, and which elements have been revealed, which were active, and which were selected). \\
\indent (iii) \textbf{Almighty} adversary, which knows upfront the outcomes of all random events, which includes the realization of $R(x)$ and the outcome of the random bits that the OCRS might query. 
\end{definition}
We are now ready to define the notion of $c$-selectability.
\begin{definition}\label{c-selectability_def} ($c$-selectability) Let $c \in [0,1]$. An OCRS for $P$ is $c$-selectable against an adversary $\mathcal{A}$ if for any $x \in P$ and $e \in N$, we have 
$$\Pr[e \text{ is selected by the OCRS against $\mathcal{A}$}| e \text{ is active}] \geq c$$
\end{definition}
It is often true that a larger probability of selection $c$ can be achieved when $x$ is supposed to be in a down-scaled version of $P$. 
\begin{definition}($(b,c)$-selectability). Let $b,c \in [0,1]$. An OCRS for $P$ is $(b,c)$-selectable against an adversary $\mathcal{A}$ if for any $x \in b \cdot P$ and $e \in N$, we have 
$$\Pr[e \text{ is selected by the OCRS against $\mathcal{A}$}| e \text{ is active}] \geq c$$
\end{definition}

An important observation is that a $(b,c)$-selectable OCRS for $b \cdot P$ against $\mathcal{A}$ implies a $bc$-selectable OCRS for $P$ against $\mathcal{A}$.
\begin{observation}\label{(b,c)-sel_implies_bc_sel} (\cite{fsz16})
A $(b,c)$-selectable OCRS for $P$ implies a $bc$-selectable OCRS for $P$.
\end{observation}
The reduction in Observation \ref{(b,c)-sel_implies_bc_sel} is as follows: The $bc$-selectable OCRS essentially runs the given $(b,c)$-selectable OCRS while scaling down by $b$ each of the probabilities $x_i$ online (i.e. consider selecting an active element independently with probability $b$). For more details see \cite{fsz16}.
\begin{remark}\label{rmk_c_vs_(b,c)}
    It is important to emphasize that a $c$-selectable ORCS for $P$ gives selection guarantees on all $x \in P$, while a $(b,c)$-selectable OCRS for $P$ gives guarantees only when $x \in b \cdot P$ (i.e. a scaled-down version of $P$). 
\end{remark}

\subsection{OCRS and Prophet inequalities}
As previously discussed, one of the many applications of OCRS is to the prophet inequality problem. Here we define the general setting of the prophet inequality problem. We begin with a setup for the environment.\\

\noindent\textbf{General setting and prophet: } We are given a group set $N = \{e_1, \ldots, e_n\}$ and a downwards-closed family $\mathcal{F} \subseteq 2^N$ of feasible subsets. Each of the elements $e_i$ is associated with a value $v_i$. A prophet is an offline algorithm, which sees the vector $(v_1, \ldots, v_n)$ and outputs the feasible set $\text{MAX}(v) = \text{argmax}_{I \in F} \sum_{i \in I} v_i$. We denote by $\text{OPT}(v) = \sum_{i \in \text{MAX}(v)} v_i$ the weight of the maximum set.

\begin{definition}\label{prophet_ineq_def} (Prophet inequality) Suppose we are given a downwards-closed feasibility constraint $\mathcal{F} \subseteq 2^{N}$. Suppose each element $e_i \in N$ takes value $v_i \in \mathbb{R}_{\geq 0}$ independently from some known distribution $\mathcal{D}_i$. These values are presented one-by-one to an online algorithm $\pi$ in an adversarial order (again specified as offline, online, or almighty). On seeing a value the algorithm needs to immediately and irrevocably decide whether to select the next element $e_i$, while always maintaining that the set of selected elements so far is in $\mathcal{F}$. Let's denote the set of selected elements by $\pi$ as $A^{*}(v)$. We say that $\pi$ induces a prophet inequality with competitive ratio $c$ for $\mathcal{F}$ if 
$$\mathbb{E}_{v \sim \mathcal{D}}[ \sum_{i \in A^{*}(v)}  v_i] \geq c \cdot \mathbb{E}_{v \sim \mathcal{D}}[ \text{OPT}(v)]$$
where $\mathcal{D} = \mathcal{D}_1 \times \ldots \times \mathcal{D}_n$ is the product of the independent distributions $\mathcal{D}_i$
\end{definition}

As with OCRS, in order to talk about $\alpha$-approximation, we need to specify the power of the adversary which specifies the order of the elements revealed. Completely analogously to Definition \ref{adv_stength_ocrs}, we could have \textit{offline}, \textit{online}, and \textit{almighty} adversaries. \\
\indent \cite{fsz16} showed that a $c$-selectable OCRS against a particular adversary $\mathcal{A}$ implies a $c$-approximation prophet inequality against an adversary of the same strength. 
\begin{theorem}\label{ocrs_implies_prophet}(\cite{fsz16})
A $c$-selectable OCRS against (offline/online/almighty) adversary implies the existence of a $c$-approximation prophet inequality algorithm. 
\end{theorem}

\indent In the classical prophet inequality formulation~\cite{KrengelS78}, the value of the online algorithm $\pi$ is compared directly to the offline optimum.~\cite{ls18} consider an ex-ante prophet inequality, where the value of $\pi$ is compared to the optimal value of a convex relaxation, which upper bounds the offline optimum.~\cite{ls18} show that this stronger notion of an ex-ante prophet inequality is equivalent to an OCRS.

\subsection{$k$-uniform matroids}\label{OCRS_prophet_ineq_section}

In this section, we give a definition for $k$-uniform matroids, which is the feasibility constraint that we will use throughout the paper. 
Given a ground set $N = \{e_1, \ldots, e_n\}$, the $k$-uniform matroid is the matroid consisting of all subsets of $N$ of size at most $k$.
\begin{definition} ($k$-uniform matroid)
The $k$-uniform matroid for $N$ is $M_k = (N,\mathcal{F}_k)$ , where $$\mathcal{F}_k = \{S \subseteq N| |S| \leq k\}$$
and the corresponding polytope of $\mathcal{F}_k$ is given by
$$P_{k} = \{x \in \mathbb{R}_{\geq 0}^n| \sum_{i=1}^n x_i \leq k\}$$
\end{definition}

We remind the reader of prior work on OCRSs and prophet inequalities for $k$-uniform matroids below.

\begin{theorem}[\cite{hks07,Alaei11,AKW14,jmz22}] The following is known, prior to our work, on OCRSs and prophet inequalities for $k$-uniform matroids:
\begin{itemize}
\item Against a fixed-order/online adversary, the best prophet inequalities and OCRSs for $\mathcal{F}_k$ achieve a guarantee of $1-\Theta(\frac{1}{\sqrt{k}})$ (lower bound:~\cite{hks07}, algorithm:~\cite{Alaei11,jmz22}).
\item Against an almighty adversary, the best prophet inequalities for $\mathcal{F}_k$ achieve a guarantee of $1-\Theta(\frac{1}{\sqrt{k}})$ (lower bound:~\cite{hks07}, algorithm:~\cite{AKW14}).
\item Against an almighty adversary, the best-known OCRS achieves a guarantee of $1-\Theta(\sqrt{\frac{\log k}{k}})$ (implicit in ~\cite{hks07}). 
\end{itemize}
\end{theorem}

\section{Warmup: Naive approaches towards an OCRS}\label{naive_algorithms}

The goal of this section is to explore a few exceptionally simple algorithms that one might try to use to construct an optimal OCRS for $k$-uniform matroids. We will present results about whether optimal factors are possible against adversaries of variable strengths. In particular, we will show that one cannot achieve a $(1-O(\frac{1}{\sqrt{k}}))$-selectable OCRS by using a very simple greedy algorithm or a tweaked variant of it utilizing a partition matroid. 

We first consider a naive greedy OCRS, which greedily selects active elements until it has selected $k$ elements. Formally in the language of Definition \ref{greedy_ocrs_def}, for this OCRS we have $\mathcal{F} = \{S \subseteq N | |S| = k \}$ and $\mathcal{F}_x = \mathcal{F}$ for all $x$. We quickly establish that the naive greedy OCRS is not $(b,c)$-selectable even against an offline/fixed-order adversary for any $b,c$ satisfying $bc = 1-O(\frac{1}{\sqrt{k}})$. This would rule out constructing exceptionally simple OCRS via Observation \ref{(b,c)-sel_implies_bc_sel}. 

\begin{theorem}\label{no_bc}
There are no $b,c$, satisfying $bc = 1-O(\frac{1}{\sqrt{k}})$, such that the naive greedy OCRS is $(b,c)$-selectable against the offline adversary. 
\end{theorem}
\begin{proof}
See Appendix \ref{appendix_naive_algorithms} for a proof. 
\end{proof}

Our second result is that even if we complicate the naive greedy OCRS slightly it does not imply an optimal factor. Suppose instead of using the $k$-uniform matroid, we use a partition matroid and the algorithm is to greedily select active elements as long as the set of selected elements lies in the partition matroid. Formally in the language of Definition \ref{greedy_ocrs_def}, the feasibility family $\mathcal{F}_{x}$ is given by a partition matroid (which could depend on $x$). For a given $x$, the partition matroid is of the form $\{(n_i,k_i,S_i)\}_{i=1}^s$, where $\sum_{i=1}^s k_i = k$, $\sum_{i=1}^s n_i = n$ and the sets $S_i$ are pairwise disjoint and satisfy $\cup_{i=1}^{s} S_i = N$ and $|S_i| = n_i$. Here we can select at most $k_i$
elements from $S_i$. We next prove that such scheme is not $(b,c)$-selectable even against the \textit{offline} adversary for any $b,c$ satisfying $bc = 1-O(\frac{1}{\sqrt{k}})$. 
\begin{theorem}\label{partition_thoerem}
There are no $b,c$, satisfying $bc = 1-O(\frac{1}{\sqrt{k}})$, such that the naive greedy OCRS with a partition matroid is $(b,c)$-selectable against the offline adversary.
\end{theorem}
\begin{proof}
See Appendix \ref{appendix_naive_algorithms} for a proof. 
\end{proof}
\begin{remark}
    Theorems \ref{no_bc} and \ref{partition_thoerem} say that one cannot use the transformation in Observation \ref{(b,c)-sel_implies_bc_sel} on a naive greedy OCRS to obtain a $(1-O(\frac{1}{\sqrt{k}}))$-selectable OCRS. Thus ruling out some exceptionally simple ways to construct a $(1-O(\frac{1}{\sqrt{k}}))$-selectable OCRS for $k$-uniform matroids.
\end{remark}
\begin{remark}
    Intuitively, the above variations of the naive greedy schemes fail to be optimal because they tend to select too many elements early in the process in comparison to the expected number of active elements so far. Thus, they are likely to run out of space when they reach the last element. Our main algorithm in Section \ref{simple_ocrs_section} attempts to counter this. 
\end{remark}

\vspace{1mm}

We conclude by reminding the reader that the naive greedy OCRS is $(b,c)$-selectable against the almighty adversary for some $b,c$ satisfying $bc= 1-O(\sqrt{\frac{\log(k)}{k}})$. The proof is implicit in the analysis of the $(1-O(\sqrt{\frac{\log(k)}{k}}))$-approximate prophet inequality of~\cite{hks07}.
\begin{theorem}[Implicit in \cite{hks07}]\label{sqrt_lgk_ocrs}
The naive greedy OCRS is $(1-\sqrt{\frac{2 \log(k)}{k}},1-\frac{1}{k})$-selectable against the almighty adversary. By the tranformation in Observation \ref{(b,c)-sel_implies_bc_sel} this implies a $(1-O(\sqrt{\frac{\log(k)}{k}}))$-selectable OCRS for $k$-uniform matroids.
\end{theorem}
\begin{proof} 
We remind the reader of the simple proof in Appendix \ref{appendix_naive_algorithms}. 
\end{proof}

\section{A simple optimal OCRS for $k$-uniform matroids}\label{simple_ocrs_section}

The goal of this section is to give a new $(1-O(\frac{1}{\sqrt{k}}))$-selectable OCRS against offline/fixed-order adversaries. Because the adversary must commit to an ordering using just knowledge of $x$, and the distribution of $R(x)$, we let $e_1,e_2 \ldots, e_n$ refer to the elements that are revealed, in order. Note that the events $e_i \in R(x)$ are independent. We will show that the following algorithm is a $(1-O(\frac{1}{\sqrt{k}}))$-selectable OCRS for $k$-uniform matroids.\\

\noindent \textbf{OCRS($x$)}
\begin{enumerate}
\item[1.] Initialize the set of selected elements $A_0 = \emptyset$.
\item[2.] For $i = 1, \ldots, n$ do: 
\begin{enumerate}
\item[2.a] If $e_i$ is active and $|A_{i-1}| +1 \leq (1-\frac{1}{\sqrt{k}})(\sum_{j \leq i} x_j) + \sqrt{k}$, then select $e_i$ with probability $(1-\frac{1}{\sqrt{k}})$ (i.e. $A_i = A_{i-1} \cup e_i$) and otherwise discard it (i.e. $A_i = A_{i-1}$).  
\item[2.b] If $|A_{i-1}| +1 > (1-\frac{1}{\sqrt{k}})(\sum_{j \leq i} x_j) + \sqrt{k}$ or $e_i$ is inactive, then discard $e_i$ (i.e. $A_i = A_{i-1}$).
\end{enumerate}
\end{enumerate}

Observe that \textbf{OCRS} is derived from an even simpler $(1-\frac{1}{\sqrt{k}},1-O(\frac{1}{\sqrt{k}}))$-selectable OCRS, by using the reduction of Observation~\ref{(b,c)-sel_implies_bc_sel}. We clearly state this simpler algorithm below (we also state it parameterized by $d$, as our entire analysis follows for general $d$, and then is optimized for $d:=\sqrt{k}$ at the very end).\\

\noindent\textbf{Algorithm($d,x$)}
\begin{enumerate}
\item[1.] Initialize the set of selected elements $B_0 = \emptyset$.
\item[2.] For $i = 1, \ldots, n$ do: 
\begin{enumerate}
\item[2.a] If $e_i$ is active and $|B_{i-1}| +1 \leq \sum_{j \leq i} x_j +d$, then select $e_i$, and otherwise discard it.  
\item[2.b] If $|B_{i-1}| +1 > \sum_{j \leq i} x_j + d$ or $e_i$ is inactive, then discard $e_i$.
\end{enumerate}
\end{enumerate}

\begin{observation}\label{obs:main} If \textbf{Algorithm}($\sqrt{k},x$) is a $(1-\frac{1}{\sqrt{k}},1-O(\frac{1}{\sqrt{k}}))$-selectable OCRS for $P_k$, then \textbf{OCRS}($x$) is a $(1-O(\frac{1}{\sqrt{k}}))$-selectable OCRS for $P_k$.
\end{observation}
\begin{proof}
\textbf{OCRS}($x$) is exactly the result of applying the~\cite{fsz16} reduction of Observation~\ref{(b,c)-sel_implies_bc_sel} to \textbf{Algorithm}($\sqrt{k},x$), with $b = 1-\frac{1}{\sqrt{k}}$ and $c = $. Therefore, by Observation~\ref{(b,c)-sel_implies_bc_sel}, \textbf{OCRS}($x$) is $(1-\frac{1}{\sqrt{k}})\cdot(1-O(\frac{1}{\sqrt{k}}))$-selectable (i.e.~$(1-O(\frac{1}{\sqrt{k}}))$-selectable) whenever \textbf{Algorithm}($\sqrt{k},x$) is $(1-\frac{1}{\sqrt{k}},1-O(\frac{1}{\sqrt{k}}))$-selectable.\end{proof}

In line with Remark \ref{rmk_c_vs_(b,c)}, we emphasize that \textbf{Algorithm}($\sqrt{k},x$) operates in a universe where the probabilities $x_i$ are scaled down by $b = 1-\frac{1}{\sqrt{k}}$, while \textbf{OCRS}($x$) operates with the original probabilities $x_i$.

Our key proposition is that \textbf{Algorithm}($d$,$x$) is indeed sufficiently selectable.

\begin{proposition}\label{prop:main} \textbf{Algorithm}($d$,$x$) is $(1-\frac{d}{k},1-\frac{2}{d-1})$-selectable over $P_k$. That is, for all $x \in (1-\frac{d}{k})\cdot P_k$, \textbf{Algorithm}($d$,$x$) is $(1-\frac{2}{d-1})$-selectable.
\end{proposition}

The proof of Proposition~\ref{prop:main} proceeds in two steps. The first (shorter) step is to guarantee that \textbf{Algorithm}($d,x$) always selects at most $k$ elements. The second is to show that every element is accepted with sufficient probability.

\begin{observation}\label{obs:stepone}
For all $x \in (1-\frac{d}{k})\cdot P_k$, \textbf{Algorithm}($d,x$) accepts at most $k$ elements.
\end{observation}
\begin{proof}
Observe that, at all times, $|B_i| \leq \sum_{j \leq i} x_j +d$. As $\sum_j x_j \leq k \cdot (1-\frac{d}{k}) = k-d$, this implies that $|B_n| \leq k$, and the algorithm accepts at most $k$ total elements.
\end{proof}\\

We now proceed to prove that the algorithm is sufficiently selectable. For this part of the analysis, it will be convenient to consider the process $S_i := |B_i| - \sum_{j \leq i} x_j$. Observe that $S_i$ has the following dynamics. First, $S_0 = 0$. Further if $S_{i-1} + 1-x_i \leq d$ (i.e. have "space" to accept $e_i$), then 
$$S_i :=\begin{cases}
        S_{i-1}+ 1-x_i & \text{if $e_i$ is active (with probability $x_i$)}\\
        S_{i-1}-x_i & \text{if $e_i$ is inactive (with probability $1-x_i$)}\\
        \end{cases}
        $$
and if $S_{i-1} + 1-x_i > d$ (i.e. don't have "space" to accept $e_i$), then 
$S_{i} = S_{i-1} -x_i$ regardless of whether $e_i$ is active (with probability 1). Additionally, consider the process $W_i = \sum_{j=1}^{i}X_j$, where 
$$
X_i =\begin{cases}
         1-x_i & \text{if $e_i$ is active (with probability $x_i$)}\\
        -x_i & \text{if $e_i$ is inactive (with probability $1-x_i$)}\\
        \end{cases}
        $$
Intuitively, $W_i$ tracks the number of \textit{active} elements above expectation (among the first $i$), and $S_i$ tracks the number of \emph{selected} elements above the expected number of active elements (among the first $i$). 
\begin{lemma}\label{prop_diff}
$W_i-S_i$ is exactly equal to the number of active elements that are discarded amongst the first $i$. 
\end{lemma}
\begin{proof}
Let $G_i$ be the set of active elements in the first $i$ revealed elements. Then, by definition $W_i = |G_i| - \sum_{j \leq i} x_j$. Thus, $W_i-S_i = |G_i| - |B_i|$ and the conclusion follows because $|B_i|$ is the number of active and selected elements in the first $i$ revealed ones. 
\end{proof}

 By Lemma \ref{prop_diff} it follows that the difference $W_i-S_i$ increases by one if and only if $e_i$ is active and discarded, and stays the same otherwise. The rest of the proof requires just two more natural steps. First, we characterize which elements are active and discarded, just as a property of the random process $W$. Second, we bound the probability of this occurring.
 
\begin{lemma}\label{integral_height} Element $e_i$ is active and discarded by \textbf{Algorithm}($d,x$) if and only if $\ceil{W_i} > d$ and $\ceil{W_i} > \ceil{W_j}$ for all $j < i$. That is, $e_i$ is active and discarded by \textbf{Algorithm}($d,x$) if and only if the random process $W$ reaches a new integral height for the first time.
\end{lemma}
 \begin{proof}
By Lemma \ref{prop_diff} the difference $W_{i}-S_{i}$ increases by 1 when $e_i$ is active and discarded and stays the same otherwise. Therefore, the $n$-th active and discarded elements is $e_{i}$, where $i$ is the smallest index such that $W_{i}-S_{i} = n$. Let's now fix an arbitrary $n \geq 1$. To prove the lemma, it is enough to show that $i$ is the smallest index such that $W_i-S_i = n$ if and only if $i$ is the smallest index such that $\ceil{W_i} = d+n$.

We begin with the ``only if'' direction. Suppose $i_n$ is the smallest index such that $W_{i_n}-S_{i_n} = n$. By Lemma \ref{prop_diff}, element $e_{i_n}$ is active and discarded. Therefore, by the selection rule of \textbf{Algorithm}($d,x$) we get that $S_{i_{n}-1} + 1-x_{i_n} >d$ and therefore $S_{i_n} = S_{i_n-1}-x_{i_n}$. This implies that $S_{i_n} > d-1$. We also know that $S_{i_n}\leq d$ (because $S_i$ stores the difference between $|B_i|$ and $\sum_{j \leq i} x_i$, which is hard-coded to be at most $d$). Using these two inequalities along with our hypothesis that $W_{i_n}-S_{i_n} = n$ we obtain:
$$W_{i_n} = S_{i_n}+n \in (d+n-1, d+n],\text{ and therefore: }\ceil{W_{i_n}} = d+n.$$
 Further, by definition of $i_n$ as the \emph{first} index such that $W_i - S_i = n$, we know that for all $j < i_n$: $W_{j} \leq S_j + n-1$. Moreover, we also know that $S_j \leq d$ for all $j$ (again, because $S_j$ stores the difference between $|B_j|$ and $\sum_{\ell \leq j} x_\ell$, which is at most $d$).
Therefore, $W_j \leq d+n-1$, and we conclude that $\ceil{W_j} \leq d+n-1 < \ceil{W_{i_n}}$ for all $j < i_n$. Therefore, $i_n$ is the smallest index such that $\ceil{W_{i_n}} = d+n$. This establishes that if $i_n$ is the smallest index such that $W_{i_n}-S_{i_n} = n$, then $i_n$ is the smallest index such that $\ceil{W_{i_n}} = d+n$. 

 Now we show the ``if'' direction. Suppose that $i_n$ is the smallest index such that $\ceil{W_{i_n}} = d+n$. Since $W_{i_n}-S_{i_n}$ is an integer, and because $S_{i_n} \leq d$, it must be the case that $W_{i_n}-S_{i_n} \geq n$ . Let $i^{*} \leq i_n$ be the smallest index such that $W_{i^{*}}-S_{i^{*}} = n$ (such $i^{*}$ exists because $W_i-S_i$ always increases by 1 or stays the same, and because we have just shown that $W_{i_n}-S_{i_n} = n$). By the ``only if'' proof above, $i^{*}$ is the smallest index such that $\ceil{W_{i^{*}}} = d+n$, implying that in fact $i^* = i_n$, as desired. This establishes that if $i_n$ is the smallest index such that $\ceil{W_{i_n}} = d+n$, then $i_n$ is also the smallest index such that $W_{i_n}-S_{i_n}=n$.

 This completes the proof: the $n^{th}$ active element discarded by the algorithm is the smallest index such that $W_i - S_i = n$. By the work above, this is exactly the smallest index such that $\ceil{W_i} = n+d$. Therefore, discarded elements are exactly those that reach a new integral height for the first time.
\end{proof} 

Our remaining task is simply to upper bound the probability that $W_i$ reaches a new integral height, for all $i$.

\begin{lemma}\label{prob_upper_bd} For each $m \in [1,n]$
$$\Pr(\ceil{W_m} > d \text{ and } \ceil{W_{m}} > \ceil{W_i} \text{ for $i < m$} ) \leq \frac{2x_m}{d-1}$$
\end{lemma}
\begin{proof}
Fix $m \in [1,n]$. Consider the process $Q_i := W_{m-i-1}-W_{m-1}$ for $i = 0,1, \ldots, m-1$. Intuitively, $\{Q_i\}_{i=0}^{m-1}$ is the ``reversed'' process $W$ starting at time-step $m-1$. Note that $Q_0 = 0$ and $Q_{m-1} = -W_{m-1}$. Observe also that
$$Q_{i}-Q_{i-1} = W_{m-i-1}-W_{m-1}-(W_{m-i}-W_{m-1}) = W_{m-i-1}-W_{m-i} = -X_{m-i}$$
Note that since the adversary does not see which elements are active and has to commit to their order a priori, we know that $X_1, \ldots, X_n$ are independent. We also note that 
$$\mathbb{E}[X_i] = x_i(1-x_i) -x_i(1-x_i) = 0$$
By combining the previous two facts, we obtain that $\{Q_{i}\}_{i=0}^{m-1}$ is a discrete martingale. Let's denote its maximum by $M_{m-1} = \max_{1 \leq i \leq m-1} Q_i$. We will next show that if $\ceil{W_m} > d$ and $\ceil{W_{m}} > \ceil{W_i}$ for $i < m$, then the following two events have to hold:
\begin{itemize}
    \item $e_m$ is active.
    \item $M_{m-1} < 1$ and $Q_{m-1} \leq -(d-1)$. That is, the martingale $Q$ can never reach a height of $1$, \emph{and} it must finish below $-(d-1)$. 
\end{itemize}
Indeed, since $\ceil{W_{m}} > \ceil{W_{m-1}}$, then $e_m$ is active. Also since $\ceil{W_m} > \ceil{W_i}$ for $i < m$, we know that $W_m > W_{m-i-1}$ for all $i$. Combining this with the inequality $W_m \leq W_{m-1}+1$, we obtain 
$$Q_i = W_{m-i-1}-W_{m-1} \leq W_{m-i-1}-W_m+1 < 1$$
Thus, $Q_i < 1$ for all $i \in [1,m-1]$ or equivalently $M_{m-1} < 1$. The condition $\ceil{W_m} > d$ implies that $W_m > d$. Using the last inequality we obtain $$Q_{m-1} = -W_{m-1} \leq -W_m+1 \leq -(d-1)$$
Therefore, we obtained that $e_m$ is active, $M_{m-1} < 1$ and $Q_{m-1} \leq -(d-1)$ as desired. 

Using the above property combined with the fact that the event whether $e_m$ is active is independent of the events $M_{m-1} < 1$ and $Q_{m-1} \leq -(d-1)$ (because $M_{m-1}$ and $Q_{m-1}$ are determined entirely by the previous $m-1$ elements) we obtain that:
\begin{align*}
&\Pr(\ceil{W_m} > d \text{ and } \ceil{W_{m}} > \ceil{W_i} \text{ for $i < m$} ) \\
&\leq \Pr(\{e_m \text{ is active}\} \cap \{M_{m-1} < 1 \text{ and }Q_{m-1} \leq -(d-1)\})\\
&=\Pr(e_m \text{ is active}) \cdot \Pr(M_{m-1} < 1 \text{ and }Q_{m-1} \leq -(d-1))\\
&= x_m \cdot \Pr(M_{m-1} < 1 \text{ and }Q_{m-1} \leq -(d-1))
\end{align*}

So, the final step is to upper bound $\Pr(M_{m-1} < 1 \text{ and }Q_{m-1} \leq -(d-1))$, which is just a claim about martingales that change by at most $1$ in each step. Lemma~\ref{lemma_martingale_bounds}, which is a short application of the Optional Stopping Theorem, applied for $a = 1, b = -(d-1) < 0$, and $K = 1$ implies that: 
$$\Pr(M_{m-1} < 1 \text{ and }Q_{m-1} \leq -(d-1)) \leq \frac{2}{d-1}.$$
Therefore, 
$$\Pr(\ceil{W_{m}} > \ceil{W_i} \text{ for $i < m$ and } \ceil{W_m} > d) \leq \frac{2x_m}{d-1},$$
which concludes the proof of the lemma. 
\end{proof}

This suffices to wrap up the proof of Proposition~\ref{prop:main}.

\begin{proof}[Proof of Proposition~\ref{prop:main}]
We have that 
\begin{align*}
    \Pr(e_m \text{ is selected} | e_m \text{ is active}) &= 1-\Pr(e_m \text{ is discarded} | e_m \text{ is active}) \\
    &= 1-\frac{\Pr(e_m \text{ is active and discarded })}{\Pr(e_m \text{ is active})} \\
    &\geq 1-\frac{\frac{2x_m}{d-1}}{x_m} = 1-\frac{2}{d-1} \text{ (by Lemma \ref{prob_upper_bd}).}
\end{align*}\end{proof}

Setting $d = \sqrt{k}$ in Proposition \ref{prop:main} we get
\begin{corollary}\label{cor:main} \textbf{Algorithm}($\sqrt{k},x$) is $(1-\frac{1}{\sqrt{k}},1-\frac{2}{\sqrt{k}-1})$-selectable, and therefore \textbf{OCRS}($x$) is $(1-O(\frac{1}{\sqrt{k}}))$-selectable.
\end{corollary}

\begin{remark}
An intuitive comparison with the extremely simple OCRS implied in \cite{hks07} is the following. The OCRS from~\cite{hks07} needs to "scale" probabilities down with $(1-\Theta(\sqrt{\frac{\log(k)}{k}}))$ in order to have a $(1-\Theta(\sqrt{\frac{\log(k)}{k}}))$ chance of not running out of space. Our result shows that if we adaptively select active elements whenever the number of selected ones does not exceed the expected number by $d$, then we only need to "scale" down probabilities by $(1-\Theta(\frac{1}{\sqrt{k}}))$ to have a $(1-\Theta(\frac{1}{\sqrt{k}}))$ chance of not running out of space.
\end{remark}

\begin{remark}
    Intuitively, an online adversary can potentially manipulate the algorithm by revealing an element $e_j$ with a high with large (resp. small) $x_j$ when the algorithm has selected many (resp. few) elements above the expectation. In this scenario, the proof of Lemma \ref{prob_upper_bd} will not go through since the "reversed" process $Q_i$ could have correlated steps and will (in general) fail to be a martingale. 
\end{remark}

\section{Upper bound against almighty adversaries}\label{opt_guarantee_against_almigthy}
Recall that \cite{hks07} implies that the naive greedy OCRS can be used to obtain a $(1-O(\sqrt{\frac{\log(k)}{k}}))$-selectable OCRS against the almighty adversary. (Theorem \ref{sqrt_lgk_ocrs}) We recall that the almighty adversary knows which elements are active a priori and also knows everything about the OCRS (Definition \ref{adv_stength_ocrs}). In this section we show that, against the almighty adversary, the probability of selection of any OCRS cannot be greater than $(1-\Omega(\sqrt{\frac{\log(k)}{k}}))$. This implies that, against almighty adversaries the factor of $(1-O(\sqrt{\frac{\log(k)}{k}}))$, achieved by the naive greedy OCRS, is asymptotically optimal. In particular, our main goal in this section is to prove the following theorem.

\begin{theorem}\label{thm_ocrs_almigthy}
Suppose that a $c$-selectable OCRS for the $k$-uniform matroid against the almighty adversary exists. Then $c \leq 1- \Omega(\sqrt{\frac{\log(k)}{k}})$
\end{theorem}

In our proof of Theorem \ref{thm_ocrs_almigthy} we will only consider the instance $n = 2k$ and vector of probabilities $x_i = \frac{1}{2}$ for all $i \in [1,2k]$ (i.e. each element is active with probability $\frac{1}{2}$). To execute the proof, we will use the following strategy. 
\begin{itemize}
    \item We will consider the following subclass of almighty adversaries. The adversary knows which elements are active and everything about the OCRS. However, it needs to: 
   \begin{itemize}
       \item commit to the order in which the elements will be revealed a priori 
       \item reveal all active elements before all inactive ones
   \end{itemize}
   This type of restriction will be convenient for the analysis. We will refer to the class of such adversaries as $\mathcal{H}$. 
   \item Assuming the existence of a $c$-selectable OCRS $\pi$ against adversaries of class $\mathcal{H}$, we show that there exists $c$-selectable OCRS $\pi^{s}$ against adversaries of class $\mathcal{H}$, which selects the $i$-th revealed active element with probability independent of the identities of the first $i$ revealed elements. In other words, the probability that $\pi^{s}$ selects the $i$-th revealed active element is a function $g(i)$ (which depends on $\pi$ but not on the identities of the first $i$ revealed elements). (Section \ref{defining_symmetric_ocrs_section})
   \item By using the probability values $\{g(i)\}_{i=1}^{2k}$, we
   construct an adversary in $\mathcal{H}$. Based on this adversary, we upper bound the probability of selection $c$ by using the solution to a linear programming relaxation. (Section \ref{adversary_section})
\end{itemize}

Before we proceed with the proof we will fully describe a general model for how an OCRS works. An arbitrary OCRS $\pi$ operates in the following way:\\
\indent \textbf{1.} Before any elements are revealed $\pi$ can flip some random coins $coins_1$. \\
\indent \textbf{2.} Once the first element $a_1$ is revealed (and whether it is active or not), $\pi$ can flip more random coins $coins'_1$ and it makes a decision to select / discard $a_1$ with some probability based on $coins_1$, $coins'_1$ and the identity of the element $a_1$. Let $b_1$ be the indicator random variable of the event that $a_1$ is selected.\\
\indent \textbf{3.} Based on the history so far, i.e. $coins_1,a_1,coins'_1, b_1$, it flips more random coins $coins_2$ before the second element is revealed. \\
\indent \textbf{4.} After the identity of the second element $a_2$ (and its activity) is revealed $\pi$ flips more coins $coins'_2$ (based on the history so far), and makes a decision to select $a_2$ or not, which is recorded in an indicator variable $b_2$.\\
\indent \textbf{5.} Let $B_i = (coins_{i}, a_i, coins'_{i}, b_i)$ where $coins_i$ is the random coins $\pi$ flips right before seeing the $i$-the revealed element $a_i$, $coins'_i$ is the random coins $\pi$ flips after seeing $a_i$, and $b_i$ is the indicator of the event that $\pi$ selects $a_i$. \\
\indent \textbf{6.} In general, as a function of the history $ \overline{B}_{i} = \{B_1, \ldots, B_i\}$, $\pi$ flips random coins $coins_{i+1}$. Then the $(i+1)$-th element $a_{i+1}$ is revealed. Based on the history $\overline{B}_i \cup coins_{i+1} \cup a_{i+1}$, $\pi$ flips more random coins $coins'_{i+1}$. Finally, based on the history $\overline{B}_{i} \cup coins_{i+1} \cup a_{i+1} \cup coins'_{i+1}$, $\pi$ decides to select / discard $a_{i+1}$, which generates the indicator $b_{i+1}$. \\
\indent \textbf{7. } The procedure described in \textbf{6.} is repeated until the $n$-th element is selected / discard.

\subsection{Defining a symmetric OCRS}\label{defining_symmetric_ocrs_section}
Suppose we have a $c$-selectable OCRS $\pi$. Our goal in this section will first be to define what it means to ``apply a permutation'' to $\pi$. Then we will define the ``symmetric'' OCRS $\pi^{s}$ in the following way:\\
\indent \textbf{1} Sample uniformly random permutation $\sigma$ of $N$.\\
\indent \textbf{2} Apply $\sigma$ to $\pi$.\\
We will then show that the probability that $\pi^{s}$ selects the $i$-th active revealed element is independent on the identities of the first $i$ revealed elements. \\
\indent We will first introduce some definitions. Given an OCRS $\pi$, and a permutation $\sigma$ of the ground set $N$, we define $\pi_{\sigma}$ as the OCRS which ``treats'' each element $a_i$ exactly like $\pi$ would ``treat'' $\sigma^{-1}(a_i)$. Formally we use the following definition.
\begin{definition}\label{permutation_ocrs_def} Let $\pi$ be an OCRS and $\sigma$ a permutation of $N$. Define $\pi_{\sigma}$ as an OCRS which uses $\pi$ as a black-box in the following way: \\
\indent \textbf{1.} Before any elements are revealed $\pi_{\sigma}$  queries $\pi$ to flip some random coins $coins_1$. \\
\indent \textbf{2.} Once the first element $a_1$ is revealed (and whether it is active or not), $\pi_{\sigma}$ queries $\pi$ on $(coins_1, \sigma^{-1}(a_1))$ to flip more random coins $coins'_1$. Based on $(coins_1, \sigma^{-1}(a_1), coins'_1)$, $\pi$ will make some decision to select / discard $\sigma^{-1}(a_1)$ with some probability. Then $\pi_{\sigma}$ selects $a_1$ if and only if $\pi$ selects $\sigma^{-1}(a_1)$. Let $b_1$ be the indicator random variable of the event that $a_1$ is selected by $\pi_{\sigma}$.\\
\indent \textbf{3.} $\pi_{\sigma}$ queries $\pi$ on history
$coins_1,\sigma^{-1}(a_1),coins'_1, b_1$ to generate random coins $coins_2$ before the second element is revealed.\\
\indent \textbf{4.} After the identity of the second element $a_2$ (and its activity) is revealed $\pi_{\sigma}$ queries $\pi$ on
$coins_1,$ $\sigma^{-1}$, $(a_1),coins'_1, b_1, coins_2, \sigma^{-1}(a_2)$ to flip more random coins $coins'_2$. Then it queries $\pi$ on $coins_1,$ $\sigma^{-1}(a_1),$ $coins'_1, b_1, coins_2, \sigma^{-1}(a_2), coins'_2$ whether to select / discard $\sigma^{-1}(a_2)$, and $\pi_{\sigma}$ selects $a_2$ if and only if $\pi$ selects $\sigma^{-1}(a_2)$. \\
\indent \textbf{5.} Let $B_i = (coins_{i}, a_i, coins'_{i}, b_i)$ where $coins_i$ is the random coins $\pi_{\sigma}$ flips by querying $\pi$ right before seeing the $i$-the revealed element $a_i$, $coins'_i$ is the random coins $\pi_{\sigma}$ flips by querying $\pi$ after seeing $a_i$, and $b_i$ is the indicator of the event that $\pi_{\sigma}$ selects $a_i$. \\
\indent \textbf{6.} In general, as a function of the history $ \overline{B}_{i} = \{B_1, \ldots, B_i\}$, $\pi_{\sigma}$ will query 
$\pi$ on $ \overline{B'}_{i} = \{B'_1, \ldots, B'_i\}$, where 
$B'_i = (coins_i, \sigma^{-1}(a_i), coins'_{i+1}, b_i)$ to flips random coins $coins_{i+1}$ before the $(i+1)$-th active element $a_{i+1}$ is revealed. Then it queries $\pi$ on $\overline{B'}_i \cup coins_{i+1} \cup \sigma^{-1}(a_{i+1})$, to generate more random coins $coins'_{i+1}$, and then it queries it again on $\overline{B'}_i \cup coins_{i+1} \cup \sigma^{-1}(a_{i+1}) \cup coins'_{i+1}$ whether to select $a_{i+1}$ if and only if $\pi$ selects $\sigma^{-1}(a_{i+1})$. This generates an indicator $b_{i+1}$. \\
\indent \textbf{7. } The procedure described in \textbf{6.} is repeated until the $n$-th element is selected / discarded. 
\end{definition}
Similarly to applying a permutation to an OCRS $\pi$, we can apply a permutation to an adversary $\mathcal{A}$, resulting in an adversary $\mathcal{A}_{\sigma}$. Intuitively, $\mathcal{A}_{\sigma}$ ``treats'' element $a$ like $\mathcal{A}$ would treat $\sigma^{-1}(a)$. We give the following formal definition. Recall that $\mathcal{H}$ is the class of adversaries defined in the beginning of this section. 
\begin{definition}\label{permutation_adv_def} Let $\mathcal{A} \in \mathcal{H}$ be an adversary and $\sigma$ a permutation of the ground set $N$. We define the adversary $\mathcal{A}_{\sigma}$ as operating against an OCRS $\pi$ in the following way:\\
\indent \textbf{1. } Given a set of active elements $A$, $\mathcal{A}_{\sigma}$ queries $\mathcal{A}$ on a set of active elements $\sigma^{-1}(A)$ against the OCRS $\pi_{\sigma^{-1}}$. Upon this query $\mathcal{A}$ returns an order in which to reveal the elements from $N$.\\
\indent \textbf{2. } Given this order, if $\mathcal{A}$ chose to reveal element $a$ in the $i$-th position, $\mathcal{A}_{\sigma}$ reveals $\sigma(a)$ in the $i$-th position. 
\end{definition}
\begin{remark}
It is easy to see that if $\mathcal{A} \in \mathcal{H}$, then $\mathcal{A}_{\sigma} \in \mathcal{H}$. 
\end{remark}
We next show that the operation of applying a permutation to an OCRS or adversary is invertible. 
\begin{proposition}\label{prop_bijection}
For a given permutation $\sigma$, the maps $\pi \to \pi_{\sigma}$ and $\mathcal{A} \to \mathcal{A}_{\sigma}$ are bijections, with inverses $\pi \to \pi_{\sigma^{-1}}$ and $\mathcal{A} \to \mathcal{A}_{\sigma^{-1}}$ respectively. 
\end{proposition}
\begin{proof}
See Appendix \ref{appendix_opt_upper_bound} for a proof. 
\end{proof}

We will now need the following helpful lemma. 
\begin{lemma}\label{parallel_permutation}
Let $a \in N$ be an element and $A \subseteq N$ a subset of elements. Let $\sigma$ be a permutation of $N$, $\pi$ an OCRS, and $\mathcal{A} \in \mathcal{H}$ an adversary. Then
$$\Pr(\pi \text{ selects $a$ against $\mathcal{A}$}|A \text{ are active}) = \Pr(\pi_{\sigma} \text{ selects $\sigma(a)$ against $\mathcal{A}_{\sigma}$}|\sigma(A) \text{ are active})$$
\end{lemma}
\begin{proof}
Consider the interaction between $\pi_{\sigma}$ and $\mathcal{A}_{\sigma}$ on set of active elements $\sigma(A)$. By Definition \ref{permutation_adv_def} before the process begins $\mathcal{A}_{\sigma}$ queries $\mathcal{A}$ on $\sigma^{-1}(\sigma(A)) = A$ against OCRS $(\pi_{\sigma})_{\sigma^{-1}} = \pi$ (Proposition \ref{prop_bijection}). Based on this if $\mathcal{A}$ chooses to first reveal the active elements $A$ to $\pi$ in some order, $\mathcal{A}$ will reveal $\sigma(A)$ to $\pi_{\sigma}$ in the same order (Definition \ref{permutation_adv_def}). Thus, if we now pair $\pi$ and $\pi_{\sigma}$ as in Definition \ref{permutation_ocrs_def}, we know that in the interaction when $\pi$ is given history $\overline{B}$ and $\pi_{\sigma}$ will be given history $\overline{B'}$, obtained from $\overline{B}$ by replacing each element $b$ by $\sigma(b)$. Thus, by Definition \ref{permutation_ocrs_def} it follows that $\pi_{\sigma}$ will select element $\sigma(a)$ if and only if $\pi$ selects element $a$. This finishes the proof. 
\end{proof}

We are now ready to show if $\pi$ is $c$-selectable against adversaries in $\mathcal{H}$, then $\pi_{\sigma}$ is also $c$-selectable against adversaries in $\mathcal{H}$.
\begin{lemma}\label{permutation_ocrs_cselectable}
If $\pi$ is $c$-selectable against adversaries in $\mathcal{H}$, then $\pi_{\sigma}$ is also $c$-selectable against adversaries in $\mathcal{H}$. 
\end{lemma}
\begin{proof}
Let $\mathcal{A} \in \mathcal{H}$ be an adversary and $a \in N$ an arbitrary element. We have that 
\begin{align*}
    &\Pr(\pi_{\sigma} \text{ selects $a$ against $\mathcal{A}$})= \\
    &= \sum_{A \subseteq N} \Pr(\pi_{\sigma} \text{ selects $a$ against $\mathcal{A}$}|\text{$A$ are active}) \Pr(\text{$A$ are active}) \\
    &=\sum_{A \subseteq N} \Pr(\pi_{\sigma} \text{ selects $a$ against $\mathcal{A}$}|\text{$A$ are active}) \frac{1}{2^{2k}} \\
    &= \sum_{A \subseteq N} \Pr(\pi\text{ selects $\sigma^{-1}(a)$ against $\mathcal{A}_{\sigma^{-1}}$}|\text{$\sigma^{-1}(A)$ are active}) \frac{1}{2^{2k}} \text{ (Lemma \ref{parallel_permutation} and Prop. \ref{prop_bijection})} \\
    &= \sum_{A \subseteq N} \Pr(\pi\text{ selects $\sigma^{-1}(a)$ against $\mathcal{A}_{\sigma^{-1}}$}|\text{$\sigma^{-1}(A)$ are active}) \Pr(\text{$\sigma^{-1}(A)$ are active})\text{ ($x_i = \frac{1}{2}, \forall i$)}\\
    &=  \Pr(\pi\text{ selects $\sigma^{-1}(a)$ against $\mathcal{A}_{\sigma^{-1}}$}) \geq \frac{c}{2}\text{ (by $c$-selectability of $\pi$} )
\end{align*}

Thus,  $$\Pr(\pi_{\sigma} \text{ selects $a$ against $\mathcal{A}$}| \text{ $a$ is active}) = \frac{\Pr(\pi_{\sigma} \text{ selects $a$ against $\mathcal{A}$})}{\Pr(\text{$a$ is active})} \geq c$$ since $\Pr(\text{$a$ is active})= \frac{1}{2}$. 
\end{proof}
\begin{remark}
Notice that in the proof of Lemma \ref{permutation_ocrs_cselectable} it was crucial that $x_i = \frac{1}{2}$ for all $i$, which we used to claim $\Pr(\text{$A$ are active}) = \Pr(\text{$\sigma^{-1}(A)$ are active}) = \frac{1}{2^{2k}}$.
\end{remark}
Suppose we have a $c$-selectable OCRS $\pi$ against adversaries in $\mathcal{H}$. We define the ``symmetric'' OCRS $\pi^{s}$ in the following way:\\
\indent \textbf{1.} Sample a uniformly random permutation $\sigma$ of the ground set $N$.\\
\indent \textbf{2.} Operate like $\pi_{\sigma}$\\
We will next show that $\pi^{s}$ is $c$-selectable against adversaries in $\mathcal{H}$ and that it does not differentiate between identities of different elements. 
\begin{lemma}\label{lemma_symmetric_c_slectable}
The OCRS $\pi^{s}$ is $c$-selectable against adversaries in $\mathcal{H}$.
\end{lemma}
\begin{proof}
Let $\mathcal{A} \in \mathcal{H}$ be an adversary. Since $\mathcal{A}$ needs to decide on the order in which to reveal the elements before seeing what permutation $\sigma$ is drawn in step \textbf{1.}, we know that the order of elements revealed does not depend on $\sigma$. Thus, for an arbitrary element $a$ we have that 
\begin{align*}
    \Pr(\pi^{s} \text{ selects $a$ against $\mathcal{A}$}) &= \sum_{\sigma} \Pr(\pi_{\sigma} \text{ selects $a$ against $\mathcal{A}$})  \frac{1}{n!} \geq \text{(by Lemma \ref{permutation_ocrs_cselectable})}\\
    &\geq \sum_{\sigma} \frac{c}{2} \frac{1}{n!} = \frac{c}{2}
\end{align*}
as desired. \end{proof}

We will now state the key lemma for this Section. Namely, that $\pi^{s}$ selects the $i$-th revealed active element with probability independent of the identities of the first $i$ revealed elements. 

\begin{lemma}\label{lemma_prob_is_independent_of_order}
Against adversaries in $\mathcal{H}$, the probability that $\pi^{s}$ selects the $i$-th revealed element, conditioned on it being active, is given by a function $g(i)$ that is independent of the identities of the first $i$ revealed elements and their order. 
\end{lemma}
\begin{proof}
Consider any adversary $\mathcal{A} \in \mathcal{H}$. By definition $\mathcal{A}$ decides on the order in which to reveals the elements apriori, and reveals all active elements before all inactive ones. Suppose that the first $i$ elements that $\mathcal{A}$ reveals are $a_1, \ldots, a_i$ in that order. To prove the lemma it is enough to show that $$\Pr(\pi^{s} \text{ selects $a_i$}|a_1, \ldots, a_i \text{ are revealed}) = g(i)$$
where $g$ is allowed to depend only on $\pi$. Note that the permutation $\sigma$ drawn is independent of the order of the elements $a_1, \ldots, a_i$ by definition of $\mathcal{H}$. Thus, we have 
\begin{align*}
   &\Pr(\pi^{s} \text{ selects $a_i$}|a_1, \ldots, a_i \text{ are revealed}) =\\
   &= \sum_{\sigma}  \Pr(\pi_{\sigma} \text{ selects $a_i$}|a_1, \ldots, a_i \text{ are revealed})\Pr(\sigma \text{ is drawn}|a_1, \ldots, a_i \text{ are revealed})\\
&= \sum_{\sigma}\Pr(\pi_{\sigma} \text{ selects $a_i$}|a_1, \ldots, a_i \text{ are revealed}) \frac{1}{n!} \\
    &= \frac{1}{n!} \sum_{\sigma}\Pr(\pi_{\sigma} \text{ selects $a_i$}|a_1, \ldots, a_i \text{ are revealed}) \\
    &= \frac{1}{n!} \sum_{\sigma}\sum_{B \subseteq \{a_1, \ldots, a_{i-1}\}}\Pr(\pi_{\sigma} \text{ selects $a_i$} \text{ and } \pi_{\sigma} \text{ selected $B$}|a_1, \ldots, a_i \text{ are revealed}) \text{ (Definition \ref{permutation_ocrs_def})}\\
    &= \frac{1}{n!} \sum_{\sigma}\sum_{B \subseteq \{a_1, \ldots, a_{i-1}\}}\Pr(\pi\text{ selects $\sigma^{-1}(a_i)$}\text { and }\pi \text{ selected $\sigma^{-1}(B)$} | \sigma^{-1}(\{a_1, \ldots, a_i\}) \text{ revealed}) \text{    ($\Delta$)}
\end{align*}
In the second inequality we used the fact that the event that $\sigma$ is drawn is independent of the decision of the adversary for which $a_1, \ldots, a_i$ to reveal. The key observation is that expression $(\Delta)$ does not depend on the elements $a_1,\ldots, a_i$ but only on $i$.  Notice that for fixed $A',B',a'$, such that $|A'| = i$, $B' \subset A'$ and $a' \in A' \setminus B'$, there are exactly 
$$\binom{i-1}{|B'|}(|B'|)!(i-1-|B'|)!(n-i)!$$ terms in the sum ($\Delta$) of the form 
$$\Pr(\pi\text{ selects $a'$}\text { and }\pi \text{ selected $B'$}| A' \text{ revealed})$$
To see this consider the number of $(\sigma, B)$ which are solutions to $$\sigma^{-1}(a_i) = a', \sigma^{-1}(B) = B', \sigma^{-1}(A) = A'$$ where $A = \{a_1, \ldots, a_i\}$. There are $\binom{i-1}{|B'|}$ ways to choose $B$. Given $B$ there are $(|B'|)!(i-1-|B'|)!(n-i)!$ ways to choose $\sigma$ in order to send $a'$ to $a_i$, $B'$ to $B$, and $A'$ to $A$. Therefore, $\Delta$ is equal to 
$$\sum_{|A'| = i,  B' \subset A' , a' \in A' \setminus B'} \Pr(\pi\text{ selects $a'$}\text { and }\pi \text{ selected $B'$}| A' \text{ revealed}) \binom{i-1}{|B'|}(|B'|)!(i-1-|B'|)!(n-i)!$$
which only depends on $i$. Therefore, the probability that $\pi^{s}$ selects the $i$-th revealed element given that it is active only depends on $i$ we will denote it by $g(i)$. \end{proof}

\subsection{Upper bound on the probability of selection}\label{adversary_section}
In this section we will present an adversary in the class $\mathcal{H}$ and show how this adversary implies the upper bound of $1-\Omega(\sqrt{\frac{\log(k)}{k}})$ on $c$. By Lemma \ref{lemma_prob_is_independent_of_order} it follows that, against adversaries in $\mathcal{H}$, the probability that $\pi^{s}$ selects the $i$-th revealed active element is given by a function $g(i)$. \\
\indent We now describe the adversary. We will only specify what the adversary does when $e_1$ is active. \\

\textbf{Adversary $\mathcal{A}^{*}$: }\\
\indent \textbf{1. } If there are $m$ active elements except for $e_1$, the adversary computes $g(1), \ldots, g(m+1)$. \\
\indent \textbf{2. } Before the process starts, the adversary finds $j = \argmin_{i \in [1,m+1]} g(i)$ and reveals element $e_1$ at position $j$ and all other $m$ active elements on positions $j' \in [1,j-1] \cup[j+1,m+1]$ in arbitrary order. \\

It is not hard to see that the adversary $\mathcal{A}^{*}$ is in the class $\mathcal{H}$ because it commits to the order apriori and reveals all active elements before all inactive. We will next show the following lemma for the probability that $\pi^{s}$ selects $e_1$ against the above adversary.
\begin{lemma}\label{prob_sele1_adversary}
Let $h(m+1) = \min_{i \in [1,m+1]}g(i)$ for $m \in [0,2k-1]$. Then
$$\Pr(\pi^{s} \text{ selects $e_1$ against $\mathcal{A}^{*}$}| e_1 \text{ is active}) = \sum_{i=1}^{2k}h(i) \frac{\binom{2k-1}{i-1}}{2^{2k-1}}$$
\end{lemma}
\begin{proof}
Suppose $e_1$ is active. Notice that the number of active elements is equal to $N+1$, where $N \sim Bin(2k-1, \frac{1}{2})$. By definition of $\mathcal{A}^{*}$, when there are $m$ active elements (except for $e_1$), the probability that $\pi^{s}$ selects $e_1$ is equal to $h(m+1)$. Therefore, by the law of total probability, $\pi^{s}$ selects $e_1$ with probability 
\begin{align*}
\Pr(\pi^{s} \text{ selects $e_1$ against $\mathcal{A}^{*}$}| e_1 \text{ is active}) &= \sum_{i=1}^{2k} h(i) \Pr[N+1 = i]\\
& =\sum_{i=1}^{2k} h(i) \Pr[Bin(2k-1,\frac{1}{2}) = i-1] \\
&= \sum_{i=1}^{2k} h(i) \frac{\binom{2k-1}{i-1}}{2^{2k-1}}
\end{align*}
as desired.
\end{proof}

We will next show a property on the values $\{g(i)\}_{i=1}^{2k}$ that will be useful later. 

\begin{lemma}\label{sum_of_gi}
$$\sum_{i=1}^{2k} g(i) \leq k$$
\end{lemma}
\begin{proof}
Suppose that all $2k$ elements in $N$ are active. We know that any adversary in $\mathcal{H}$ will choose an order for the elements to be revealed before $\pi^{s}$ draws $\sigma$. In that case we know by Lemma \ref{lemma_prob_is_independent_of_order} that $\pi^{s}$ will select the $i$-th revealed element with probability $g(i)$. Let $I_{i}$ be the indicator random variable that $\pi^{s}$ selects that $i$-th revealed element. Note that $E[I_{i}] = g(i)$. Since $\pi^{s}$ never selects more than $k$ elements (Definition \ref{ocrs_def}) we have that 
$$k \geq  \mathbb{E}[\sum_{i=1}^{2k} I_i] = \sum_{i=1}^{2k} \mathbb{E}[I_i] = \sum_{i=1}^{2k} g(i)$$
as desired. 
\end{proof}\\

As a last step towards Theorem \ref{thm_ocrs_almigthy}, we consider a linear program relaxation whose optimal objective upper bounds the probability of selection $c$ (Lemma \ref{linear_relax_upper}), and characterize an optimal solution to the program (Lemma \ref{lp_solution_property}). Finally, we show that the optimal objective of the linear program is at most $1-\Omega(\sqrt{\frac{\log(k)}{k}})$ (Lemma \ref{anti_concentration_ineq}).
\begin{lemma}\label{linear_relax_upper}
Let $c^{*}$ denote the optimal value of the following linear program 
\begin{equation}\label{linear_relaxation}
\begin{array}{ll@{}ll}
\text{maximize}  & \displaystyle\sum\limits_{i=1}^{2k} f(i) \frac{\binom{2k-1}{i-1}}{2^{2k-1}} &\\
\text{subject to}& \displaystyle\sum_{i=1}^{2k}f(i) \leq k \\
&f(i) \geq f(i+1)\text{ for } & i=1 ,\dots, 2k-1\\
&f(2k)\geq 0\\
\end{array}
\end{equation}
Then if $\pi^{s}$ is $c$-selectable it holds that
$$c \leq c^{*}$$
\end{lemma}
\begin{proof}
We first claim that $\{h(i)\}_{i=1}^{2k}$ is a feasible assignment to (\ref{linear_relaxation}). Notice that by definition we have that $h(i) \leq g(i)$ for $i \in [1,2k]$. Using this combined with Lemma \ref{sum_of_gi} we obtain 
$$\sum_{i=1}^{2k}h(i) \leq \sum_{i=1}^{2k}g(i) \leq k$$
Further, note that by definition we have that $h(i) \geq h(i+1)$ for $i \in [1,2k-1]$ and clearly $h(2k) \geq 0$. Combining the aforementioned observations we get that the vector $\{h(i)\}_{i=1}^{2k}$ is a feasible assignment of (\ref{linear_relaxation}). Therefore, by Lemma \ref{prob_sele1_adversary} we obtain that 
$$c^{*} \geq \sum_{i=1}^{2k}h(i) \frac{\binom{2k-1}{i-1}}{2^{2k-1}} = 
\Pr(\pi^{s} \text{ selects $e_1$ against $\mathcal{A}^{*}$}| e_1 \text{ is active})$$
By Lemma \ref{lemma_symmetric_c_slectable} we know that $\pi^{s}$ is $c$-selectable i.e. 
$$\Pr(\pi^{s} \text{ selects $e_1$ against $\mathcal{A}^{*}$}| e_1 \text{ is active}) \geq c$$
By combining the above two inequalities we obtain that $$c^{*} \geq c$$
which finishes the proof.
\end{proof}\\
We will next prove a claim for the optimal solution of the linear program (\ref{linear_relaxation}). 
\begin{lemma}\label{lp_solution_property}
The optimal solution of (\ref{linear_relaxation}) has the form $f(i) = x$  for $i =1 \ldots, k+a$, $f(k+a+1) = y$, and $f(i) = 0$ for $i > k+a+1$ for some $x \geq y \geq 0$ and $a \in [0,k]$. 
\end{lemma}

\begin{proof}
First, note that the constraints of the linear program (\ref{linear_relaxation}) define a bounded convex polytope, so there exists a feasible assignment that achieves the optimum of (\ref{linear_relaxation}). Let $b_i = \frac{\binom{2k-1}{i-1}}{2^{2k-1}}$, we know that $b_i < b_{i+1}$ for $i <k$, $b_{k} = b_{k+1}$, and $b_{i} > b_{i+1}$ for $i > k$. Let $f^{*}$ be an optimal solution to (\ref{linear_relaxation}). Suppose that $f^{*}(i) > f^{*}(i+1)$ for some $i < k$. Then consider $f^{**}$ defined by $f^{**}(j) = f^{*}(j)$ for $j \not \in \{i,i+1\}$, and $f^{**}(i) = f^{**}(i+1) = \frac{f^{*}(i) + f^{*}(i+1)}{2}$. Note that $f^{**}$ still has decreasing non-negative entries and the sum of its entries is equal to that of $f^{*}$ i.e. it is feasible. The difference between the objective values of $f^{**}$ and $f^{*}$ is equal to 
$$\frac{f^{*}(i) + f^{*}(i+1)}{2}(b_i +b_{i+1}) - b_i f^{*}(i) -b_{i+1}f^{*}(i+1) = \frac{(f^{*}(i+1)-f^{*}(i))(b_i-b_{i+1})}{2} > 0$$
since $f^{*}(i+1) < f^{*}(i)$ and $b_i < b_{i+1}$. Thus, a contradiction with the optimality of $f^{*}$. Therefore, $f^{*}(i) = f^{*}(1)$ for all $i \leq k$. \\
\indent Let $f^{*}(i) = x$ for $i \leq k$. If $f^{*}(j) \in \{x,0\}$ for $j > k$ we are done. Otherwise let $j > k$ be the smallest index such that $x > f^{*}(j)> 0$, and let $f^{*}(j) = y$. This means $f^{*}(i) = x$ for $i < j$. Assume that $f^{*}(j+1)  > 0$. Then, let $l > j$ be the largest index such that $f^{*}(l) > 0$, and let $f^{*}(l)  = z$. Choose $\epsilon < \min(x-y,z)$ and consider $f^{**}$ defined by $f^{**}(j) = f^{*}(j) + \epsilon$, $f^{**}(l) = f^{*}(l)-\epsilon$, and $f^{**}(i) = f^{*}(i)$ for $i \not \in \{j, l\}$. Note that $f^{**}$ is feasible by the choice of $\epsilon$ since its entries are still decreasing and have the same sum as those of $f$.  The difference between the objective of $f^{**}$ and $f^{*}$ equals to
$$(f^{*}(j) + \epsilon)b_j + (f^{*}(l)-\epsilon)b_l - f^{*}(j) b_j - f^{*}(l)b_l = \epsilon(b_j - b_l) > 0$$ 
as $b_j > b_l$ because $l > j > k$, which contradicts the optimality of $f^{*}$. Thus, we showed that $f^{*}(j+1) = 0$, which finishes the proof of the Lemma.
\end{proof}\\

Note that by Lemma \ref{lp_solution_property} we know that the optimal value of (\ref{linear_relaxation}) has the following form 
\begin{equation}\label{c_star_eq}
c^{*} = x \Pr(Bin(2k-1, \frac{1}{2}) \leq k+a-1) + y \Pr(Bin(2k-1, \frac{1}{2}) = k+a)
\end{equation}
for some $x \geq y \geq 0$ satisfying $x(k+a) +y \leq k$, where $a \geq 0$. By using these constraints we easily obtain that $y \leq x \leq \frac{k}{k+a}$. By using this inequality in equation (\ref{c_star_eq}), we obtain that 
\begin{equation}\label{c_star_upper_bound}
    c^{*} \leq \frac{k}{k+a}\Pr(Bin(2k-1, \frac{1}{2}) \leq k+a)
\end{equation}

We now show the final Lemma of this section, which provides an upper bound for the RHS of (\ref{c_star_upper_bound}). 
\begin{lemma}\label{anti_concentration_ineq} Let $a \in [0,k]$, then 
$$\frac{k}{k+a} \Pr(Bin(2k-1,\frac{1}{2}) \leq k+a) \leq 1-\Omega(\sqrt{\frac{\log(k)}{k}})$$
\end{lemma}
\begin{proof}
It is enough to show that 
\begin{equation}\label{log_sqrt_upper}
\frac{k}{k+a} \Pr(Bin(2k-1,\frac{1}{2}) \leq k+a) \leq 1-\frac{1}{100}\sqrt{\frac{\log(k)}{k}}
\end{equation}
Let's assume, for the sake of contradiction, that inequality (\ref{log_sqrt_upper}) is not true. By this assumption we have the following chain of inequalities
\begin{align*}
    1-\frac{1}{100}\sqrt{\frac{\log(k)}{k}} &< \frac{k}{k+a} \Pr( Bin(2k-1, \frac{1}{2}) \leq k+a) \\
    & \leq \min\Bigg(\frac{k}{k+a}, \Pr(Bin(2k-1, \frac{1}{2}) \leq k+a)\Bigg) \\
    & \leq \min\Bigg(\frac{k}{k+a}, \Pr(Bin(2k-2, \frac{1}{2}) \leq k+a)\Bigg) \\
\end{align*}
where in the second line we used that each of the terms in the product is at most 1 and in the third line that $Bin(2k-1, \frac{1}{2})$ stochastically dominates $Bin(2k-2, \frac{1}{2})$. Thus,
\begin{equation}\label{useful_equation}
     1-\frac{1}{100}\sqrt{\frac{\log(k)}{k}} < \min\Bigg(\frac{k}{k+a}, \Pr(Bin(2k-2, \frac{1}{2}) \leq k+a)\Bigg)
\end{equation}
By (\ref{useful_equation}), we get 
\begin{align*}
    \frac{k}{k+a} &> 1-\frac{1}{100}\sqrt{\frac{\log(k)}{k}}\\
    k &> k+a -\frac{k+a}{100}\sqrt{\frac{\log(k)}{k}}\\
    a &< \frac{k+a}{100}\sqrt{\frac{\log(k)}{k}} \leq \frac{2k}{100}\sqrt{\frac{\log(k)}{k}} = \frac{1}{50}\sqrt{k \log(k)} \text{ (since $a \leq k$)}
\end{align*}
Combining the above inequality with (\ref{useful_equation}) again we obtain 
$$\Pr(Bin(2k-2, \frac{1}{2}) < k + \frac{1}{50} \sqrt{k \log(k)}) \geq \Pr(Bin(2k-2, \frac{1}{2}) \leq k+a) > 1-\frac{1}{100}\sqrt{\frac{\log(k)}{k}}$$
Subtracting one from both sides we get 
\begin{equation}\label{tail_upper_bd}
\Pr(Bin(2k-2, \frac{1}{2}) \geq k + \frac{1}{50} \sqrt{k \log(k)}) < \frac{1}{100}\sqrt{\frac{\log(k)}{k}}
\end{equation}
We will not use the following anti-concentration for binomial distribution given in Proposition \ref{anti_conc_binomial}. For $k' \in [\frac{n}{2}, \frac{5n}{8}]$ and even $n$ we have 
\begin{equation}\label{tail_lower_bd}
    \Pr(Bin(n, \frac{1}{2}) \geq k') \geq \frac{1}{15} \exp \Big(-16n(\frac{1}{2}-\frac{k'}{n})^2 \Big)
\end{equation}
Substituting $n = 2k-2$ and $k' = k + \frac{1}{50}\sqrt{k \log(k)}$ in (\ref{tail_lower_bd}) we obtain
\begin{align*}
    \Pr(Bin(2k-2,\frac{1}{2}) \geq k + \frac{1}{50}\sqrt{k \log(k)}) &\geq \frac{1}{15} \exp \Big(-16(2k-2)\Big(\frac{1}{2}-\frac{k + \frac{1}{50}\sqrt{k \log(k)}}{2k-2}\Big)^2 \Big)\\
    &= \frac{1}{15} \exp \Big(-32(k-1)\Big(\frac{1 + \frac{1}{50}\sqrt{k \log(k)}}{2k-2}\Big)^2 \Big) \\
    &= \frac{1}{15} \exp \Big(-8\frac{(1 + \frac{1}{50}\sqrt{k \log(k)})^2}{k-1} \Big) \\
    &\geq \frac{1}{15} \exp \Big(-8\frac{( \frac{2}{25}\sqrt{k \log(k)})^2}{k} \Big) \text{ (for big enough $k$)} \\
    &= \frac{1}{15} \exp \Big(-\frac{32}{625} \log(k)
\Big) = \frac{1}{15}\frac{1}{k^{\frac{32}{625}}}
\end{align*} 
Combining the last inequality with (\ref{tail_upper_bd}) we obtain
\begin{align*}
    \frac{1}{100} \sqrt{\frac{\log(k)}{k}} &> \frac{1}{15}\frac{1}{k^{\frac{32}{625}}}\\
    \iff \frac{15}{100}\sqrt{\log(k)} &\geq k^{\frac{1}{2}-\frac{32}{625}} = k^{\frac{561}{1250}}
\end{align*}
which fails to hold for large enough $k$. Thus, we obtain a contradiction. Therefore, (\ref{log_sqrt_upper}) is true, which proves the lemma. 
\end{proof}

\begin{remark}
    The optimal solution to the LP in Lemma \ref{linear_relax_upper} turns out not have the same values $f(i)$ as implied by the asymptotically optimal OCRS from \cite{hks07} (Theorem \ref{sqrt_lgk_ocrs}). This is because the values implied by this OCRS would have an asymptotically optimal performance as opposed to exactly instance optimal. 
\end{remark}

\subsection{Proof of Theorem \ref{thm_ocrs_almigthy}}
By Lemma \ref{linear_relax_upper}, we know that 
\begin{equation}\label{piece_one}
    c \leq c^{*} 
\end{equation}
Additionally,by combining (\ref{c_star_upper_bound}) and Lemma \ref{anti_concentration_ineq} we know that 
\begin{equation}\label{piece_two}
c^{*} \leq 1-\Omega(\sqrt{\frac{\log(k)}{k}})
\end{equation}
Combining (\ref{piece_one}) and (\ref{piece_two}) we obtain 
$$c \leq 1-\Omega(\sqrt{\frac{\log(k)}{k}})$$
finishing the proof of Theorem \ref{thm_ocrs_almigthy}.

\section{Conclusion}\label{conclusions_future_work_section}
We provide a new, simple, and optimal OCRS for $k$-uniform matroids against a fixed-order adversary. In particular, our algorithm has the advantage that it is extremely simple to implement and it does not require solving a mathematical program. Our analysis connects its performance to a random walk, and follows by concluding properties of this random walk. We expect that the tools we develop in analyzing our algorithm to be of independent interest and to have a broader applicability within online stochastic optimization. 

As our second main result, we show that no OCRS for $k$-uniform matroids can be $(1-\Omega(\sqrt{\frac{\log k}{k}}))$-selectable against an almighty adversary, establishing that the simple greedy OCRS implied by~\cite{hks07} is optimal.

\bibliographystyle{alpha}
\bibliography{References}

\appendix
\section{Omitted proofs}
\subsection{Omitted proofs from Section \ref{naive_algorithms}}\label{appendix_naive_algorithms}

\newenvironment{myproofT2}{\paragraph{Proof of Theorem \ref{no_bc}}}{\hfill $\square$}

\begin{myproofT2}
Assume the contrary, i.e. the naive greedy OCRS is $(b,c)$-selectable for some $b,c$ satisfying $bc = 1-O(\frac{1}{\sqrt{k}})$.  Suppose $n = 2k$ and $x_i = \frac{b}{2}$ for all $i$. Suppose the offline adversary reveals the elements in the order $e_1, \ldots, e_n$. Then in order for $e_n$ to be selected conditioned on it being active, it must be that the OCRS has selected at most $k-1$ elements from $e_1, \ldots, e_{n-1}$. By definition of the naive greedy OCRS this implies there must be at most $k-1$ active elements from $e_1, \ldots, e_{n-1}$. The number of active elements amongst these is distributed as $Bin(2k-1, \frac{b}{2})$. Thus, we have that 
$$\Pr(Bin(2k-1, \frac{b}{2}) \leq k-1) \geq \Pr(\text{$e_n$ is selected}|\text{$e_n$ is active}) \geq c$$
Thus by assumption, 
\begin{equation}\label{min_ineq}
\min(b,\Pr(Bin(2k-1, \frac{b}{2}) \leq k-1) \geq b\Pr(Bin(2k-1, \frac{b}{2}) \leq k-1) \geq bc \geq 1- \frac{C}{\sqrt{k}}
\end{equation}
So $b \geq 1-\frac{C}{\sqrt{k}}$. Using this we get 
$$\Pr(Bin(2k-1, \frac{1}{2}-\frac{C}{2\sqrt{k}}) \geq k) \leq \Pr(Bin(2k-1, \frac{b}{2}) \geq k)$$
because for $a \geq k$ the function  $\Pr(Bin(2k-1, x) = a) = \binom{2k-1}{a} x^{a} (1-x)^{2k-1-a}$ is increasing in $ x < \frac{1}{2}$. Therefore, 
\begin{align*}
\Pr(Bin(2k-1, \frac{1}{2}-\frac{C}{2\sqrt{k}})  \geq k) &\leq \Pr(Bin(2k-1, \frac{b}{2}) \geq k)\\
&= 1 - \Pr(Bin(2k-1, \frac{b}{2}) \leq k-1) \text{ (by (\ref{min_ineq}))}\\
&\leq 1-(1-\frac{C}{\sqrt{k}}) = \frac{C}{\sqrt{k}}
\end{align*}
However, by Lemma \ref{lemma_binomial} $$\lim_{k \to \infty} \Pr(Bin(2k-1, \frac{1}{2}-\frac{C}{2\sqrt{k}}) \geq k) \geq A(C) > 0 = \lim_{k \to \infty} \frac{C}{\sqrt{k}}$$
which contradicts the above. 
\end{myproofT2}

\begin{lemma}\label{lemma_binomial}
Let $Bin(n,p)$ denote the binomial distribution with $n$ trails and probability of success $p$. Let $C > 0$. There exists some $A(C) > 0$ such that
$$ \lim_{k \to \infty} \Pr(Bin(2k-1, \frac{1}{2}-\frac{C}{2\sqrt{k}}) \geq k) \geq A(C)$$
\end{lemma}
\begin{proof}
Let $a \in [0, \sqrt{k}]$ be an integer. We claim that 

$$\binom{2k-1}{k+a}(\frac{1}{2}-\frac{C}{2\sqrt{k}})^{k+a}(\frac{1}{2}+\frac{C}{2 \sqrt{k}})^{k-a} \geq \frac{A(C)}{\sqrt{k}}$$
for some constant $A(C) > 0$ independent of $a$. Note that 
$$(\frac{1}{2}-\frac{C}{2\sqrt{k}})^{k+a}(\frac{1}{2}+\frac{C}{2 \sqrt{k}})^{k-a} = (\frac{1}{4}-\frac{C^2}{4k})^k \frac{(1-\frac{C}{\sqrt{k}})^a}{(1+\frac{C}{\sqrt{k}})^a} \geq \frac{1}{4^k}(1-\frac{C^2}{k})^k \frac{(1-\frac{C}{\sqrt{k}})^{\sqrt{k}}}{(1+\frac{C}{\sqrt{k}})^{\sqrt{k}}}$$
Note that $(1-\frac{C^2}{k})^k \to \frac{1}{e^{C^2}}$ so it is bounded by a positive constant from below. Similarly  $\frac{(1-\frac{C}{\sqrt{k}})^{\sqrt{k}}}{(1+\frac{C}{\sqrt{k}})^{\sqrt{k}}} \to \frac{1}{e^{2C}}$ is bounded below by a positive constant (depending on $C$). \\
\indent It remains to lower bound $\binom{2k-1}{k+a}\frac{1}{4^k}$. Define $H(x) = x \log(\frac{1}{x}) +(1-x)\log(\frac{1}{1-x})$. We will use the following identity for binomial coefficients which holds for $m \in [1, \frac{n}{2}]$
\begin{equation}\label{binomial_approx}
\binom{n}{m} = \sqrt{\frac{n}{2 \pi m (n-m)}} \exp(n H(\frac{m}{n}))(1 + O(\frac{1}{n} + \frac{1}{m^3} + \frac{1}{(m-n)^3}))
\end{equation}
Substituting $n = 2k-1$ and $m = k-1-a$ in (\ref{binomial_approx}) we get
\begin{equation}\label{binom_approx_here}
\binom{2k-1}{k+a} = \binom{2k-1}{k-a-1} = \sqrt{\frac{2k-1}{2 \pi (k-a-1)(k+a)}} \exp(n H(\frac{1}{2} - \frac{a + \frac{1}{2}}{2k-1}))(1 + O(\frac{1}{k}))
\end{equation}
Note that since $a \in [0,\sqrt{k}]$ we have $\sqrt{\frac{2k-1}{2 \pi (k-a-1)(k+a)}} \geq \frac{D}{\sqrt{k}}$ for some constant $D > 0$ and $1+ O(\frac{1}{k}) > \frac{1}{2}$. By Taylor's expansion we have that 
\begin{equation}\label{taylor}
H(x + \epsilon) = H(x) + \epsilon \log(\frac{1-x}{x}) - \frac{\epsilon^2}{2x(1-x)} + O(\epsilon^3)
\end{equation}
Substituting $x = \frac{1}{2}$ and $\epsilon = -\frac{a+\frac{1}{2}}{2k-1}$ in (\ref{taylor}) we get
$$(2k-1)H(\frac{1}{2} - \frac{a + \frac{1}{2}}{2k-1}) = (2k-1)H(\frac{1}{2}) + \frac{2(a+ \frac{1}{2})^2}{2k-1} + O(\frac{(a + \frac{1}{2})^3}{(2k-1)^2})$$
Note that $H(\frac{1}{2}) = log(2)$ and since $a \leq \sqrt{k}$ the lower order terms are bounded from below by a (possibly negative) constant $F$. Thus, $$(2k-1)H(\frac{1}{2} - \frac{a + \frac{1}{2}}{2k-1}) \geq (2k-1)log(2) + F$$ for some constant $F$. Substituting the above inequalities in (\ref{binom_approx_here}) we get
$$\binom{2k-1}{k+a} \frac{1}{4^k} \geq \frac{D}{2\sqrt{k}}\exp((2k-1)log(2) + F)\frac{1}{4^k} = \frac{D'}{\sqrt{k}}$$
for some constant $D' > 0$. Combining this with the lower bounds in the beginning we get that for all $a \in [0, \sqrt{k}]$ we have
\begin{equation}\label{ineq_tool}
\binom{2k-1}{k+a}(\frac{1}{2}-\frac{C}{2\sqrt{k}})^{k+a}(\frac{1}{2}+\frac{C}{2 \sqrt{k}})^{k-a} \geq \frac{A(C)}{\sqrt{k}}
\end{equation}
for some $A(C)> 0$. Using this inequality we obtain 
\begin{align*}
\Pr(Bin(2k-1, \frac{1}{2}-\frac{C}{2\sqrt{k}}) \geq k) &\geq \sum_{a \in [0, \sqrt{k}]}\binom{2k-1}{k+a}(\frac{1}{2}-\frac{C}{2\sqrt{k}})^{k+a}(\frac{1}{2}+\frac{C}{2 \sqrt{k}})^{k-a} \text{ (by (\ref{ineq_tool}))}\\
&\geq \sqrt{k}\frac{A(C)}{\sqrt{k}} = A(C)
\end{align*}
as desired. This finishes the proof of the lemma.
\end{proof}

\newenvironment{myproofT3}{\paragraph{Proof of Theorem \ref{partition_thoerem}}}{\hfill $\square$}

\begin{myproofT3}
Assume the contrary, i.e. there exists such an OCRS based on a partition matroid that is $(b,c)$-selectable for some $b,c$ satisfying $bc = 1-O(\frac{1}{\sqrt{k}})$. Let the \textit{offline} adversary commit the natural order $e_1, \ldots, e_n$. Suppose that $n = 2k$ and consider the instance $x_i = \frac{b}{2}$ for all $i$. Let the partition matroid for this $x$ be given by $\{(n_i, k_i, S_i)\}_{i=1}^s$. The OCRS greedily selects active elements with only constraint that it selects at most $k_i$ elements in $S_i$. Since $\sum_{i=1}^{s}k_i = k$ and $\sum_{i=1} n_i = 2k$, there exists some $i$ such that $n_i \geq 2k_i$. Let $e_j$ be the last element in $S_i$ according to the order above (i.e. with largest index $j$). Then if $e_{j}$ is selected, it must be the case that at most $k_{i}-1$ of the other elements of $S_i$ have been selected. Since the OCRS is greedy there must be at most $k_i-1$ active elements amongst them, because otherwise the OCRS will select the first $k_i$ active elements in $S_i$ before $e_j$ is revealed. The number of these active elements is distributed as $Bin(n_i-1, \frac{b}{2})$. Thus, we have that 
$$\Pr(Bin(n_i-1, \frac{b}{2}) \leq k_i-1) \geq \Pr(\text{$e_{j}$ is selected}|\text{$e_{j}$ is active}) \geq c$$
By our assumption we have that
$$b\Pr(Bin(n_i-1, \frac{b}{2}) \leq k_i-1) \geq bc \geq 1-\frac{C}{\sqrt{k}}$$
which implies that $b \geq 1- \frac{C}{\sqrt{k}} \geq 1- \frac{C}{\sqrt{k_i}}$
and 
$$\Pr(Bin(n_i-1, \frac{b}{2}) \leq k_i-1) \geq 1- \frac{C}{\sqrt{k}}$$
Subtracting one from both sides we obtain
$$\Pr(Bin(n_i-1, \frac{b}{2}) \geq k_i) \leq \frac{C}{\sqrt{k}}$$
Note that since $b \geq 1-\frac{C}{\sqrt{k_i}}$ and $n_i \geq 2k_i$ we get the following chain of inequalities 
\begin{align*}
\Pr(Bin(2k_i-1, \frac{1}{2}-\frac{C}{2\sqrt{k_i}}) \geq k_i) &\leq \Pr(Bin(2k_i-1, \frac{b}{2}) \geq k_i) \\
&\leq \Pr(Bin(n_i-1, \frac{b}{2}) \geq k_i) \leq \frac{C}{\sqrt{k}}
\end{align*}
Since $\frac{C}{\sqrt{k}} \leq \frac{C}{\sqrt{k_i}}$, if $k_i$ is sufficiently large, the above inequality will stop being true due to Lemma \ref{lemma_binomial}. Thus, there exists some $M$ such that $k_i < M$. Then, we have that 
$$\frac{C}{\sqrt{k}} \geq \Pr(Bin(2k_i-1, \frac{b}{2}) \geq k_i) \geq (\frac{b}{2})^{2k_i-1} \geq (\frac{b}{2})^{2M} \geq (\frac{1}{2} - \frac{C}{2\sqrt{k}})^{2M}$$
which also fails for large $k$ because the LHS converges to $0$ while the RHS stays bounded from below by $\frac{1}{4^{2M}}$. Thus we obtained a contradiction.  
\end{myproofT3}

\newenvironment{myproofT4}{\paragraph{Proof of Theorem \ref{sqrt_lgk_ocrs}}}{\hfill $\square$}

\begin{myproofT4}
Consider the naive greedy OCRS, which greedily selects active elements until it has selected $k$ elements, and suppose the adversary is \textit{almighty}. Consider $b,c$, which will be chosen later. Let $x \in b \cdot P_k$ be the input and let's suppose that $e_i \in N$ is active. Notice that if there are at most $k-1$ active elements in $R(x) \setminus \{e_i\}$, then $e_i$ will certainly be selected. Thus, we have that 
$$\Pr(e_i \text{ is selected}| e_i \text{ is active}) \geq \Pr(|R(x) \setminus \{e_i\}| \leq k-1)$$
Note that 
$$\Pr(|R(x) \setminus \{e_i\}| \leq k-1) =1 -\Pr(|R(x) \setminus \{e_i\}| \geq k)$$
Observe that $|R(x) \setminus \{e_i\}|$ is a sum of independent $\{0,1\}$ random variable corresponding to the indicators of whether each of $e_j \neq e_i$ is active. Thus, $$\mu = \mathbb{E}[|R(x) \setminus \{e_i\}|] = \sum_{j \neq i} x_j \leq bk$$
where the last inequality is because $x \in b \cdot P_k$. Thus, by the multiplicative Chernoff bound (Theorem \ref{concentration_ineq}), 
$$\Pr(|R(x) \setminus \{e_i\}| \geq k) = \Pr(|R(x) \setminus \{e_i\}| \geq (1+\frac{k}{\mu}-1)\mu) \leq e^{-\frac{(\frac{k}{\mu}-1)^2 \mu}{2 + \frac{k}{\mu}-1}} = e^{-\frac{(k-\mu)^2}{\mu+k}}$$
Note that $\frac{(k-\mu)^2}{\mu+k}$ is decreasing for $\mu < k$, and since $\mu < bk < k$ by above, we get 
$$\Pr(|R(x) \setminus \{e_i\}| \geq k) \leq e^{-\frac{(k-\mu)^2}{\mu+k}} \leq e^{-\frac{(k-bk)^2}{bk+k}} = e^{-\frac{k(b-1)^2}{b+1}} \leq e^{-\frac{k(b-1)^2}{2}}$$
Thus combining the above inequalities we get 
$$\Pr(e_i \text{ is selected}| e_i \text{ is active}) \geq 1- e^{-\frac{(b-1)^2k}{2}} = c$$
and thus $bc = b(1- e^{-\frac{(b-1)^2k}{2}})$. Taking $b = 1-\sqrt{\frac{2 \log(k)}{k}}$ we obtain
$$bc = (1-\sqrt{\frac{2 \log(k)}{k}})(1-\frac{1}{k}) = 1-O(\sqrt{\frac{\log(k)}{k}})$$
as desired. 

\end{myproofT4}

\subsection{Lemma about martingales from Section \ref{simple_ocrs_section}}
\begin{lemma}\label{lemma_martingale_bounds}
Let $\{Q_i\}_{i \in \mathbb{N}_0}$ be discrete martingale with $Q_0 =0$ and such that there exists $K > 0$ such that  $|Q_{i+1}-Q_{i}| \leq K$ a.s. for all $i \geq 0$. Let $n \in \mathbb{N}$, $a > 0$, and $b < 0$. Also let 
$M_{n} = \max_{1 \leq i \leq n} Q_i$.
 Then, 
$$\Pr( M_n < a, Q_n \leq b) \leq \frac{a+K}{-b}$$
\end{lemma}
\begin{proof}
Define the stopping time 
$$T = \inf\{i \geq 0: Q_i \geq a \text{ or } Q_i \leq b \}$$
Notice that $T$ is a stopping time adapted to the filtration of the martingale $Q$. Also note that since the martingale increments are bounded by $K$, we have $|Q_{\min(i,T)}| \leq \max(a+K, |b| + K) < \infty$. Therefore, by the Optional Stopping Theorem (Theorem \ref{ost}), $Q_{T}$ is an almost surely well-defined random variable and
$$\mathbb{E}[Q_{T}] = \mathbb{E}[Q_0] = 0$$
Therefore, we have that
\begin{align*}
    0 = \mathbb{E}[Q_{T}] &= \mathbb{E}[Q_{T} \mathbbm{1}\{Q_T \geq a\}+Q_{T} \mathbbm{1}\{Q_T \leq b\}]\\
    \implies \mathbb{E}[Q_{T} \mathbbm{1}\{Q_T \geq a\}] &= \mathbb{E}[-Q_{T} \mathbbm{1}\{Q_T \leq b\}] 
\end{align*}
Notice that since the martingale has bounded increments by $K$, we have 
\begin{equation}\label{1}
Q_{T} \mathbbm{1}\{Q_T \geq a\} \leq (a+K)\mathbbm{1}\{Q_T \geq a\}
\end{equation}
and
\begin{equation}\label{2}
-Q_{T} \mathbbm{1}\{Q_T \leq b\} \geq -b \mathbbm{1}\{Q_T \leq b\}
\end{equation}
By (\ref{1}) and (\ref{2}), we get 
$$(a+K)\Pr(Q_T \geq a) \geq E[Q_{T} \mathbbm{1}\{Q_T \geq a\}] =  E[-Q_{T} \mathbbm{1}\{Q_T \leq b\}]  \geq -b \Pr(Q_T \leq b)$$
Therefore, since $-b > 0$ and $\Pr(Q_T \geq a) \leq 1$, we get 
$$\Pr(Q_T \leq b) \leq \frac{a+K}{-b}\Pr(Q_T \geq a) \leq \frac{a+K}{-b}$$
Now, notice that if the event $\{ M_n < a, Q_n \leq b\}$ holds then it must the case that $Q$ has reached $b$ before $a$ and this has happened at most at step $n$ i.e. $S_{T} \leq b$ and $T \leq n$. In particular, this implies that $S_{T} \leq b$. Therefore, 
$$P( M_n< a, Q_n \leq b) \leq \Pr(Q_T \leq b) \leq \frac{a+K}{-b}$$
as desired. 
\end{proof}
\subsection{Omitted proofs from Section \ref{opt_guarantee_against_almigthy}}\label{appendix_opt_upper_bound}

\newenvironment{myproofPbij}{\paragraph{Proof of Proposition \ref{prop_bijection}}}{\hfill $\square$}

\begin{myproofPbij}
For the first statement, it is enough to show that $(\pi_{\sigma})_{\sigma^{-1}} = \pi$ for all $\pi, \sigma$. By Definition \ref{permutation_ocrs_def}, when $(\pi_{\sigma})_{\sigma^{-1}}$ needs to take an action based on a particular history $\overline{B}$, it takes the action $\pi_{\sigma}$ would take on the history $\overline{B'}$ obtained from $\overline{B}$ after replacing every element in $a$ by $\sigma(a)$ (this history could also contain the current revealed element which the OCRS hasn't selected / discarded yet). Further, the action $\pi_{\sigma}$ takes on history $\overline{B'}$ is obtained by the action $\pi$ would take on history $\overline{B''}$ obtained from $\overline{B'}$ by replacing every element $a'$ by $\sigma^{-1}(a')$. Since $\sigma^{-1}(\sigma(a)) = a$, we know that $\overline{B''}$ is the same history as $\overline{B}$. Thus, $(\pi_{\sigma})_{\sigma^{-1}}$ will take the same action as $\pi$ when using the same history. This implies the first statement. \\
\indent For the second statement, it is enough to show that $(\mathcal{A}_{\sigma})_{\sigma^{-1}} = \mathcal{A}$. By Definition \ref{permutation_adv_def}, on active elements $A$ and against OCRS $\pi$, $(\mathcal{A}_{\sigma})_{\sigma^{-1}}$ queries $\mathcal{A}_{\sigma}$ on active elements $\sigma(A)$ and OCRS $\pi_{\sigma}$. Upon this $\mathcal{A}_{\sigma}$ queries $\mathcal{A}$ on active elements $\sigma^{-1}(\sigma(A)) = A$ against $(\pi_{\sigma})_{\sigma^{-1}} = \pi$ (by earlier). If $a$ is the $i$-th element of the order $\mathcal{A}$ reveals against $\pi$, by above $\sigma^{-1}(\sigma(a)) = a$ is the $i$-th element $(\mathcal{A}_{\sigma})_{\sigma^{-1}}$ reveals against $\pi$. This finishes the proof. 
\end{myproofPbij}

\section{Tools from probability theory}
\begin{theorem}\label{concentration_ineq} (Multiplicative Chernoff bound)
Suppose $X_1, \ldots, X_n$ are independent random variable taking values in $\{0,1\}$. Let $X$ denote their sum, and let $\mu = E[X]$. Then for any $\delta > 0$. We have 
$$P[X \geq (1+\delta)\mu] \leq e^{-\frac{\delta^2\mu}{2+\delta}}$$
\end{theorem}

\begin{theorem}\label{ost} (Optional Stopping Theorem)\\
Let $Q = \{Q_{i}\}_{i \in \mathbb{N}_0}$ be a discrete martingale and $T$ a stopping time. Assume that one of the following three conditions holds:\\
\indent (i) $T$ is almost surely bounded i.e. $T \leq C$ for some constant $C$.\\
\indent (ii) $\mathbb{E}[T] < \infty$ and there exists a constant $C$ such that $\mathbb{E}[|Q_{i+1}-Q_{i}|| \mathcal{F}_i] \leq C$ almost surely.\\
\indent (iii) There exists a constant $C$ such that $|Q_{\min(i,T)}| \leq C$ for all $i$.\\
Then $Q_{T}$ is an almost surely well defined random variable and $\mathbb{E}[Q_T] = E[Q_0]$
\end{theorem}

\begin{proposition}
\label{anti_conc_binomial} (Anti-concentration for Binomial distribution, Prop 7.3.2, \cite{mv08})
For $\frac{n}{2} \leq k' \leq \frac{5n}{8}$ and even $n$ we have 
$$\Pr(Bin(n,\frac{1}{2}) \geq k') \geq \frac{1}{15} \exp \Big(-16n(\frac{1}{2}-\frac{k'}{n})^2 \Big)$$
\end{proposition}

\end{document}